\documentclass{bmcart}

\usepackage{amsthm,amsmath,amssymb}
\usepackage[utf8]{inputenc} 

\def\includegraphics{}

\usepackage{graphicx}    

\startlocaldefs
\endlocaldefs

\begin{document}

\begin{frontmatter}

\begin{fmbox}
\dochead{Research}


\title{A Risk-Sensitive Task Offloading Strategy for Edge Computing in Industrial Internet of Things}


\author[
  addressref={aff1},                   
  email={xyhao18@fudan.edu.cn}   
]{\inits{X.H}\fnm{Xiaoyu} \snm{Hao}}
\author[
  addressref={aff1},                   
  email={}   
]{\inits{R.Z}\fnm{Ruohai} \snm{Zhao}}
\author[
  addressref={aff1},
  corref={aff1},                       
  email={taoyang@fudan.edu.cn}
]{\inits{T.Y.}\fnm{Tao} \snm{Yang}}
\author[
  addressref={aff3},                   
  email={}   
]{\inits{Y.H.}\fnm{Yulin} \snm{Hu}}
\author[
  addressref={aff1,aff2},              
  corref={aff1},
  email={bohu@fudan.edu.cn}            
]{\inits{B.H}\fnm{Bo} \snm{Hu}}
\author[
  addressref={aff4},                   
  email={}   
]{\inits{Y.Q}\fnm{Yuhe} \snm{Qiu}}

\address[id=aff1]{
  \orgdiv{Department of Electronics Engineering},                          
  \orgname{Fudan University},          
  \city{Shanghai},                     
  \postcode{200433},             
  \cny{China}                          
}

\address[id=aff2]{%
  \orgdiv{Key Laboratory of EMW Information (MoE)},
  \orgname{Fudan University},
  \city{Shanghai},
  \postcode{200433},
  \cny{China}
}

\address[id=aff3]{%
  \orgdiv{ISEK Research Group},
  \orgname{RWTH Aachen University},
  \city{Aachen},
  \cny{Germany}
}

\address[id=aff4]{%
  \orgdiv{China Mobile (Chengdu)},
  \orgname{Information \& Telecommunication Technology Co., Ltd.},
  \city{Chengdu},
  \cny{China}
}


\end{fmbox}


\begin{abstractbox}

\begin{abstract} 
Edge computing has become one of the key enablers for ultra-reliable and low-latency communications in the industrial Internet of Things in the fifth generation communication systems, and is also a promising technology in the future sixth generation communication systems. In this work, we consider the application of edge computing to smart factories for mission-critical task offloading through wireless links. In such scenarios, although high end-to-end delays from the generation to completion of tasks happen with low probability, they may incur severe casualties and property loss, and should be seriously treated. Inspired by the risk management theory widely used in finance, we adopt the Conditional Value at Risk to capture the tail of the delay distribution. An upper bound of the Conditional Value at Risk is derived through analysis of the queues both at the devices and the edge computing servers. We aim to find out the optimal offloading policy taking into consideration both the average and the worst case delay performance of the system. Given that the formulated optimization problem is a non-convex mixed integer non-linear programming problem, a decomposition into sub-problems is performed and a two-stage heuristic algorithm is proposed. Simulation results validate our analysis and indicate that the proposed algorithm can reduce the risk in both the queuing and end-to-end delay.
\end{abstract}


\begin{keyword}
\kwd{Ultra-reliable and low-latency communications, Edge computing, Industrial Internet of Things, Risk management theory, Conditional Value at Risk}
\end{keyword}


\end{abstractbox}
%

\end{frontmatter}



\section{Introduction}
Intelligent factory automation is one of the typical applications envisioned in ultra-reliable and low-latency communications (URLLC) scenarios in the fifth generation (5G) and the coming sixth generation (6G) communications \cite{gu2020knowledge,she2020tutorial}. In future smart factories, machines and sensors are seamlessly connected with each other through wireless links to conduct production tasks corporately. During the manufacturing process, a great number of operations of the machines and robots require complex control algorithms and intense data computation, such as travelling across zones to identify and pick up the objects and controlling the robotic arms to assemble components within a precise position alignment \cite{2018SmartFactory}. The limited built-in computing resources are not sufficient for the stringent latency requirements, so the tasks have to be offloaded to servers for processing \cite{2019EC}. Conventionally, the large volume of data generated at the local devices is uploaded to the cloud computing servers \cite{velte2009cloud}. However, since the cloud computing servers are usually deployed remotely, the large roundtrip transmission latency as well as the possible network congestion makes it hard to meet the stringent end-to-end delay requirements of the actuators and control units in the IIoT system. To overcome these difficulties, edge computing has emerged, where the servers are placed at the edge of the network to achieve a much lower transmission and processing latency \cite{2016EC_VisionandChallenge}.
\par
There have been literatures focusing on improving the service efficiency of the edge computing systems \cite{2016MEC,2019CellularEC,2019CooperativeEC,2019Adaptive,2018cloudfogcomputing}. The authors in \cite{2016MEC} adopt the Markov decision process (MDP) to minimize the average delay of a mobile edge computing system by deciding whether to compute locally or offload the tasks to the edge server. In \cite{2019CellularEC}, the authors define a delay-based Lyapunov function instead of the queue length-based one, and minimize the average delay by optimizing resource scheduling under the Lyapunov optimization framework. Three-tier multi-server mobile computing networks are investigated in \cite{2019CooperativeEC}, where a cooperative task offloading strategy is proposed based on the Alternating Direction Method of Multipliers (ADMM). In \cite{2019Adaptive}, an adaptive learning-based task offloading algorithm is proposed to minimize the average delay for vehicle edge computing systems. Most relatively, \cite{2018cloudfogcomputing} integrates the fog computing to the cloud-based industrial Internet of Things (IIoT), where the task offloading, transmission and computing resource allocation schemes are jointly optimized to reduce the service response latency in the unreliable communication environment. 
\par
The aforementioned works have only focused on the average delay performance, neglecting the worst case performance of the edge computing system. However, in the IIoT systems, the probability of an intense delay jitter usually matters much more than the average delay, since when the delay exceeds a certain threshold, severe accidents may incur such as the deadlock of the manufacture process, the damage to the machines and even casualties. Therefore, in such scenarios, not only the average delay performance, but also the potential hazard, i.e. the risk behind the tail distribution of the delay, should be carefully investigated. There have been some preliminary works to deal with the embedded risks in the edge computing systems \cite{2019Dynamic,2019FBL,2020RiskSensitiveVEC}. In \cite{2019Dynamic}, the tail distribution of the task queue under a probability constraint imposed on the excess value is characterized by the extreme value theory \cite{de2007extreme}, and an offloading strategy is designed to minimize the energy consumption. The authors of \cite{2019FBL} also apply the extreme value theory to the edge computing systems, in order to investigate the extreme event of queue length violation in the computation phase. Besides, in \cite{2020RiskSensitiveVEC}, the authors focus on a vehicular edge computing network where vehicles either fetch images from cameras or acquire synthesized images from an edge computing server. A risk-sensitive learning \cite{mihatsch2002risk} based task fetching and offloading strategy is proposed to minimize the risk behind the end-to-end delay. 
\par
Different from the works mentioned above, we introduce the risk management theory \cite{RiskManagement}, widely used in the field of finance, to the edge computing system in consideration of the uncertainty of the wireless channels. Value at risk (VaR) and Conditional Value at Risk (CVaR) are the two widely used tools to characterize risks. While VaR takes the Gaussian distribution as assumption and lacks convexity and sub-additivity, which makes it inapplicable in many cases, CVaR is a coherent risk measure of any type of probability distribution and is much easier to handle in practice. Therefore, CVaR is employed in this work to model the risk of the task completion delay in the considered edge computing-assisted IIoT system. We aim to minimize both the average delay and the CVaR by jointly designing the offloading and computing resource allocation strategy. The main contributions of this work are summarized as follows:
\begin{itemize}
  \item 
We focus on the hazard incurred by the intense delay jitter in the edge computing-assisted IIoT scenario and introduce the risk management theory to the design of the offloading and computation resource allocation strategy.
  \item 
A cascade queuing model is constructed to describe the end-to-end delay property of the system. Due to the uncertainty of the wireless channel, the transmission time follows a general distribution, which makes the queuing model hard to analyze. By exploring the queuing theory and the risk management theory, we provide an upper bound for both the average end-to-end delay and the CVaR.
  \item
A low-complexity risk-sensitive task offloading strategy is proposed, where both the average performance and the risk with respect to the end-to-end delay are optimized simultaneously. The computation complexity of each procedure of the proposed algorithm is analyzed in details. Simulations under the practical wireless environment in the automated factory validate the effectiveness of the proposed strategy in controlling the risk behind the intense delay jitter.

\end{itemize}
\par
The remainder of this paper is organized as follows. We introduce the system model and analyze both the average delay and the CVaR in Section 2. In Section 3, we formulate the offloading and computation resource allocation problem and propose a low-complexity heuristic algorithm. In Section 4, numerical results are reported with discussions. Finally, Section 5 concludes the paper.

\section{System model}
\subsection{Edge computing system}
As shown in Fig. 1, we consider an edge computing-assisted IIoT system that consists of a set of $\mathcal{M}=\{1,2,\cdots,M\}$ IIoT devices and a set of $\mathcal{N}=\{1,2,\cdots,N\}$ edge computing servers (ECS). Each IIoT device $i\in\mathcal{M}$ randomly generates tasks of identical size of $d_i$ bits, and we assume the task arrival process follows the Possion distribution with average arrival rate $\lambda_i$. We denote by $\omega_i$ the computation intensity of the task of device $i$ \cite{2017EEMEC}, i.e. the number of CPU cycles required to process per bit data. Then, the total CPU cycles needed for a task of device $i$, denoted by $c_i$, can be calculated as $c_i=\omega_i d_i$.

\begin{figure}[!t]
\includegraphics[width=3.3in]{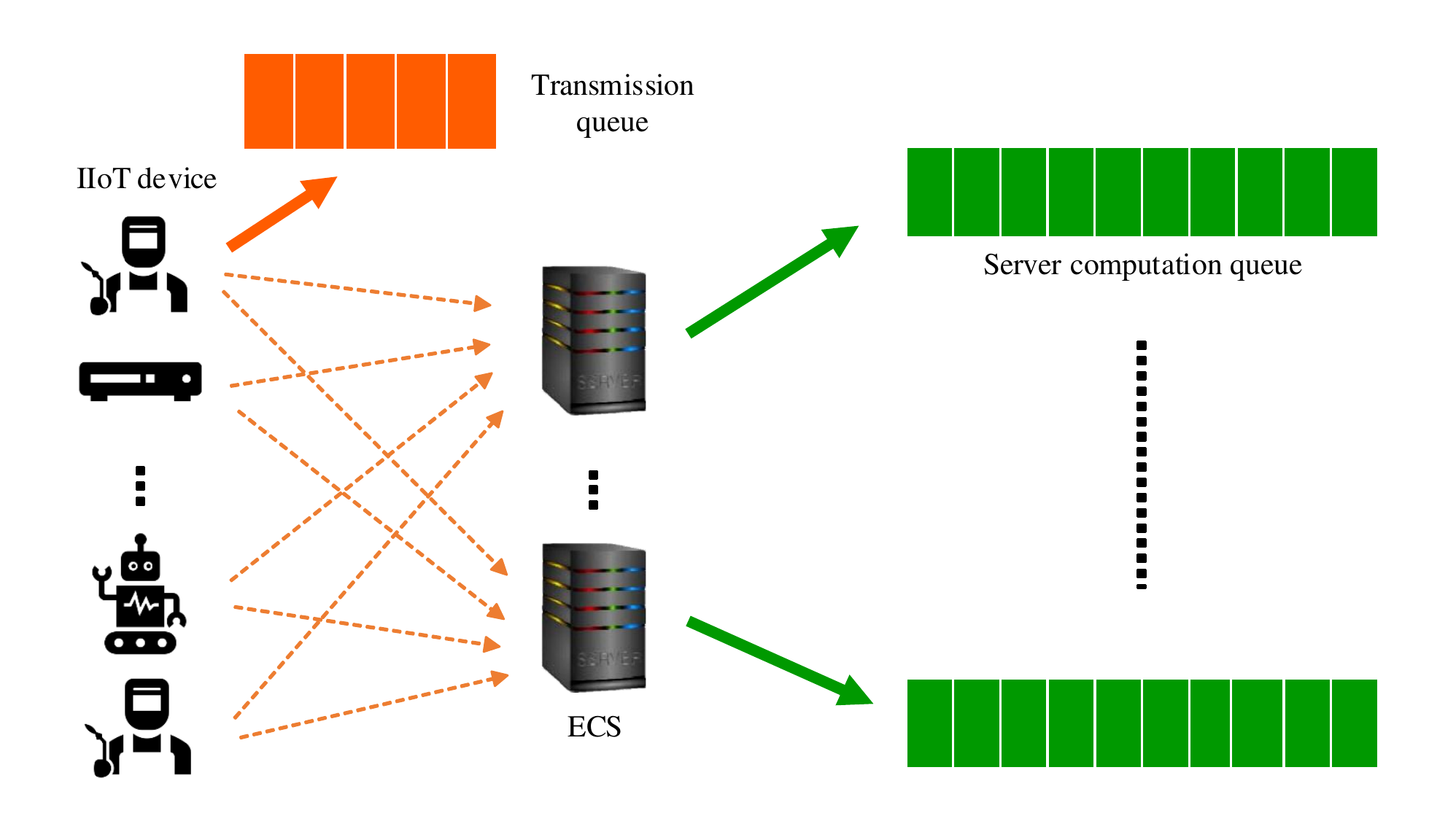}
\caption{System model}
\end{figure}

Owing to the insufficiency of computation capability, the IIoT devices offload their tasks to the ECSs through wireless links. Each ECS $j\in\mathcal{N}$ is equipped with a CPU of $N_j$ cores, which  can work simultaneously and independently. We assume that the tasks of a device can only be offloaded to one ECS, while each core only processes the tasks from the same device, which means a ECS can receive tasks from multiple devices as long as the number of devices it serves doesn't exceed the number of the CPU cores \cite{2019Dynamic}. Let $\mathbf{X}_{M \times N}=\left[ x_{ij} \right]$ as the offloading matrix, where $x_{ij}, i\in\mathbf{M}, j\in\mathbf{N}$ is defined as follows:
\begin{equation}
\label{offloadingmatrix}
x_{ij} =
\begin{cases}
1\ , & \text{device $i$ offloads its tasks to ECS $j$,} \\
0\ , & \text{otherwise.}
\end{cases}
\end{equation}
Under this definition, we can redescribe the offloading scheme in the considered system mathematically as $\sum_{j=1}^N{x_{ij}=1}$ and $\sum_{i=1}^{M}{x_{ij} \leq N_j}$ for $\forall i\in \mathbf{M}, \forall j \in \mathbf{N}$.
\par
Due to the massive deployment of devices in the complex industrial environment, it's impractical to obtain the instantaneous channel state information (CSI) accurately and timely. Therefore, in this work, we design a task offloading strategy based on statistics of the wireless links, i.e. the distribution of the channel gain. We assume blocking-fading channels such that the channel gain remains unchanged during the execution of one task and varies independently between two executions following an identical distribution, which is known a priori. Denote by $g_{ij}$ the channel gain from device $i$ to ECS $j$, the transmission rate can be expressed as follows:
\begin{equation}
\label{transmissionrate}
R_{ij} = B \log_2{( 1+ \frac{g_{ij}p_i}{N_0\Phi_{ij}})},
\end{equation}
where $B$ is the bandwidth, $N_0$ is the noise power, $p_i$ is the transmit power of device $i$ and $\Phi_{ij}$ is the path loss from device $i$ to ECS $j$. Without loss of generality, we assume that the noise power at each ECS is identical, and each IIoT device has an orthogonal channel with the same bandwidth $B$.

\subsection{Queuing model}

In the considered edge computing-assisted IIoT system, there are two kinds of queues: the queue at each device and the queue at each ECS, as depicted in Fig. 2. Without loss of generality, we assume that device $i$ offloads its tasks to ECS $j$, and denote by $Q^D_{ij}$ the queue formed at the device $i$. The arrival process of $Q^D_{ij}$ follows the Poisson process, and the departure process is dependent on the transmission delay denoted by $t^D_{ij}$, which is given by
\begin{equation}
\label{transmissiondelay}
t^D_{ij} = \frac{d_i}{R_{ij}} = \frac{d_i}{B \log_2{( 1+ \frac{g_{ij}p_i}{N_0\Phi_{ij}})}}.
\end{equation} 
Therefore, $Q^D_i$ follows the M/G/1 model. 

\begin{figure}[!t]
\includegraphics[width=3.3in]{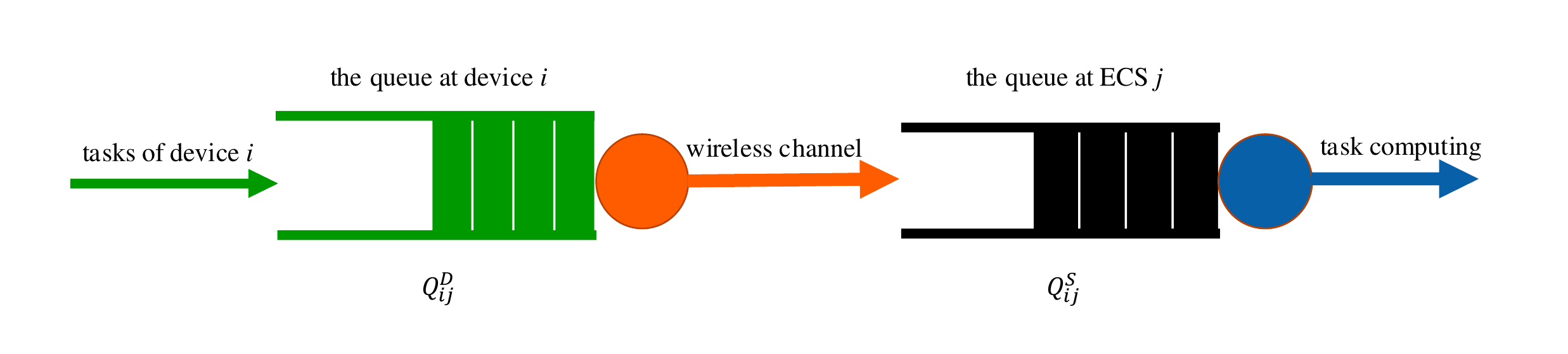}
\caption{Queuing model}
\end{figure}
\par
As for an ECS, tasks from each device connected to it form an independent queue. Denote by $Q^S_{ij}$ the queue of tasks offloaded from device $i$ to ECS $j$, then the arrival process of $Q^S_{ij}$ is the same as the departure process of $Q^D_{ij}$. We denote by the $f_{ij}$ the computation frequency allocated to device $i$ by ECS $j$, then the computation delay denoted by $t^S_{ij}$, i.e. the time required to complete a task of device $i$, is calculated as $t^S_{ij} = {c_i}/{f_{ij}}$. As a result, the service time of a task follows a deterministic distribution and thus $Q^S_{ij}$ follows the G/D/1 model.
\par
Based on the queuing analysis, when device $i$ offloads its tasks to ECS $j$, the total delay denoted by $t_{ij}$ is given by
\begin{equation}
\label{totaldelay}
t_{ij} = W^D_{ij} + t^D_{ij} + W^S_{ij} + t^S_{ij},
\end{equation}
where $W^D_{ij}$ and $W^S_{ij}$ are the queuing delay at the device $i$ and the ECS $j$, respectively. We 'll analyse both the average performance and the CVaR of the total delay in the following.

\subsection{Average delay}
According to (\ref{totaldelay}), the average delay can be calculated as follows
\begin{equation}
\label{averagedelay}
E[t_{ij}] = E[W^D_{ij}] 
+ E[t^D_{ij}] + E[W^S_{ij}] + E[t^S_{ij}].
\end{equation}
\par
As for the queuing delay at device $i$, we denote by $\mu_{ij}$ the service rate of $Q^D_{ij}$, which can be calculated as the reciprocal of the average transmission time, i.e.
\begin{equation}
\label{mu}
\mu_{ij}= \frac{1}{E[t^D_{ij}]} = \frac{d_i}{E[R_{ij}]},
\end{equation}
where the expectation is taken over the probability distribution of the channel gain. According to \cite{bertsekas1992data}, the average queuing delay at device $i$ can be expressed as follows
\begin{equation}
\label{E_W_D}
E[W^D_{ij}]=\frac{\lambda_i}{2\mu^2_{ij}(1-\rho_{ij})},
\end{equation}
where $\rho_{ij}=\lambda_i/\mu_{ij}$.
\par
To analyze the queuing delay in the G/D/1 queuing model of $Q^S_{ij}$, we first give the following lemma \cite{kleinrock1976queueing}.

\newtheorem{lemma}{Lemma}
\begin{lemma}
\label{Lemma1}
In the G/G/1 queuing model, let $\lambda$, $\mu$ and $W$ be the arrival rate, service rate and queuing delay, respectively, then an upper bound of the average queuing delay is given by
\begin{equation}
\label{UB_W}
E[W] \leq \frac{\lambda({\sigma}^2_a+{\sigma}^2_b)}{2(1-\rho)},
\end{equation}
where $\rho=\lambda/\mu$, $\sigma^2_a$ is the variance of the inter-arrival time, and $\sigma^2_b$ is the variance of the service time.
\end{lemma}
According to Lemma \ref{Lemma1}, we obtain the following theorem characterizing the upper bound of the average queuing delay at a ECS. 

\newtheorem{theorem}{Theorem}
\begin{theorem}
\label{Theorem1}
When device $i$ offloads its tasks to ECS $j$, an upper bound of the average queuing delay at ECS $j$ is given by
\begin{equation}
\label{UB_W_S}
E[W^S_{ij}] \leq  \frac{\lambda^S_{ij}{\sigma^S_{ij}}^2}{2(1-\rho^S_{ij})},
\end{equation}
where $\lambda^S_{ij}$ is the arrival rate of $Q^S_{ij}$, $\rho^S_{ij} = \lambda^S_{ij}/\mu_{ij}$ is the traffic intensity and ${\sigma^S_{ij}}^2$ is the variance of the arrival interval of tasks offloaded from device $i$ to ECS $j$.
\end{theorem}

\begin{proof}
\label{proofofTheorem1}
G/D/1 model can be seen as a special case of G/G/1 model with the service time following a deterministic distribution, the variance of which is zero. By substituting $\lambda=\lambda^S_{ij}$, $\sigma^2_a = {\sigma^S_{ij}}^2$, $\sigma^2_b = 0$ and $\rho = \rho^S_{ij}$ into (\ref{UB_W}), we get (\ref{UB_W_S}) and Theorem~\ref{Theorem1} is proved.
\end{proof}
Note that, due to the cascaded structure between $Q^D_{ij}$ and $Q^S_{ij}$, the arrival rate $\lambda^S_{ij}$ of $Q^S_{ij}$ is equal to the departure rate of $Q^D_{ij}$, which can be evaluated from the analysis of the inter-departure time in \cite{bertsimas1990departure}. Similarly, the variance of the inter-arrival time of $Q^S_{ij}$, i.e. ${\sigma^S_{ij}}^2$, can be derived from the variance of the inter-departure time of $Q^D_{ij}$ as in \cite{yeh2000characterizing}. 
\par
Finally, combining  (\ref{averagedelay}), (\ref{mu}) and (\ref{UB_W_S}), the upper bound of the average total delay can be obtained in the following corollary.

\newtheorem{corollary}{Corollary}
\begin{corollary}
\label{corollary1}
When device $i$ offloads its tasks to ECS $j$, an upper bound of the average total delay $E[t_{ij}]$ is given by
\begin{equation}
\label{UB_t}
E[t_{ij}] \leq \frac{\lambda_i}{2\mu^2_{ij}(1-\rho_{ij})} + \frac{1}{\mu_{ij}}
+ \frac{\lambda^S_{ij}{\sigma^S_{ij}}^2}{2(1-\rho^S_{ij})} + \frac{c_i}{f_{ij}}.
\end{equation}
\end{corollary}

Since each IIoT device offloads its tasks to only one ECS, for each device $i$ the task completion time denoted by $t_i$ can be expressed as $t_i=\sum^N_{j=1}{x_{ij}t_{ij}}$, and correspondingly an upper bound of the average total delay of device $i$ based on (\ref{UB_t}) is given by
\begin{equation}
\label{UB_averagetotaldelay}
E[t_i] = \sum^N_{j=1}{x_{ij}E[t_{ij}]} \leq E[t_i]^\ast,
\end{equation}
where
\begin{equation}
\label{E_t_ast}
E[t_i]^\ast = \sum^N_{j=1}{x_{ij}\left( \frac{\lambda_i}{2\mu^2_{ij}(1-\rho_{ij})} + \frac{1}{\mu_{ij}}
+ \frac{\lambda^S_{ij}{\sigma^S_{ij}}^2}{2(1-\rho^S_{ij})} + \frac{c_i}{f_{ij}} \right)}.
\end{equation}

\subsection{Risk metric for delay}
Risk in the considered edge computing-assisted system is mainly reflected in the high latency happened with low probability. Specifically, we introduce CVaR as a measure of risk to characterize the tail distribution of the delay. Before we formally define CVaR, we first give the definition of VaR \cite{jorion2000value}. 
\newtheorem{MyDef}{Definition}
\begin{MyDef}
\label{def1}
For a random variable $X$ and a confidence level $\alpha\in(0,1)$, the $\alpha$-VaR of $X$ is the $\alpha$-percentile of the distribution of $X$, which can be expressed mathematically as follows
\begin{equation}
\label{VaR}
{\rm VaR}_\alpha(X) = \inf_\gamma{\{\gamma: P(X>\gamma) \leq \alpha\}}.
\end{equation} 
\end{MyDef}
The CVaR measures the expected loss in the right tail given a particular threshold has been crossed, and can also be considered as the average of potential loss that exceed the VaR. The definition of CVaR is given as follows \cite{acerbi2002coherence}.
\begin{MyDef}
\label{def2}
For a random variable $X$ and a confidence level $\alpha\in(0,1)$, the $\alpha$-CVaR of $X$ is given by
\begin{align}
\label{CVaR}
{\rm CVaR}_\alpha(X) &= E[X|X>{\rm VaR}_\alpha(X)] \notag \\
&= \frac{1}{1-\alpha} \int_{\alpha}^1{{\rm VaR}_\theta(X)d\theta}.
\end{align}
\end{MyDef}

To characterize the CVaR of the total delay, we first analyze the CVaR of each part of the total delay in (\ref{totaldelay}). Recall that the service process of the queue at the device follows general distribution, so it is quite difficult to directly characterize the probability distribution of the waiting time. However, in the considered IIoT scenario, it is reasonable to take the \emph{heavy} traffic assumption. According to \cite{kingman1962queues}, the cumulative distribution function (CDF) of $W_{ij}^D$ can be approximated as
\begin{equation}
\label{approx_CDF_W_D}
F(W_{ij}^D) \approx 1 - {\rm exp}\left[{- \frac{2(1 - \rho_{ij})}{\lambda_i V_{ij}} W_{ij}^D}\right],
\end{equation}
where $V_{ij}$ is the variance of the service time of $Q_{ij}^D$, i.e. the transmission time of a task. Based on (\ref{approx_CDF_W_D}), the CVaR of $W_{ij}^D$ can be evaluated in the following theorem.

\begin{theorem}
\label{Theorem2}
For a confidence level $\alpha \in (0,1)$, the $\alpha$-CVaR of $W_{ij}^D$ can be expressed as
\begin{equation}
\label{CVaR_W_D}
{\rm CVaR}_\alpha(W_{ij}^D) = \frac{\lambda_i V_{ij}}{2(1- \rho_{ij})} [1 - \ln{(1-\alpha)}].
\end{equation}
\end{theorem}

\begin{proof}
\label{proofofTheorem2}
According to the definition of VaR in Definition \ref{def1}, we can obtain that
\begin{align}
\label{VaR_W_D}
{\rm VaR}_\alpha(W_{ij}^D)
&= \inf_\gamma{\{\gamma: e^{- \frac{2(1 - \rho_{ij})}{\lambda_i V_{ij}} \gamma} \leq \alpha\}} \notag \\
&= \frac{-\lambda_i V_{ij}}{2(1 - \rho_{ij})} \ln{(1-\alpha)}.
\end{align}
By substituting (\ref{VaR_W_D}) to (\ref{CVaR}), the CVaR of $W_{ij}^D$ can be calculated as follows.
\begin{align}
\label{CVaR_W_D_cal}
{\rm CVaR}_\alpha(W_{ij}^D) 
&= \frac{1}{1-\alpha} \int_{\alpha}^1{\frac{-\lambda_i V_{ij}}{2(1 - \rho_{ij})} \ln{(1-\theta)}d\theta} \notag \\
&= \frac{\lambda_i V_{ij}}{2(1- \rho_{ij})} [1 - \ln{(1-\alpha)}].
\end{align}
\end{proof}
As for the transmission delay $t_{ij}^D$, we define a auxiliary function 
\begin{equation}
\label{auxiliary}
\phi_\alpha(t_{ij}^D,\gamma) := \gamma + \frac{1}{1+\alpha}E[{(t_{ij}^D-\gamma)}^+],
\end{equation}
where $(x)^+= \max{(0,x)}$ and the expectation is taken over the distribution of the channel gain. According to \cite{rockafellar2000optimization}, the CVaR of $t_{ij}^D$ can finally be calculated as
\begin{align}
\label{CVaR_t_D_cal}
{\rm CVaR}_\alpha(t_{ij}^D) 
&= \min_{\gamma \in \mathcal{R}}{\phi_\alpha(t_{ij}^D,\gamma)} \notag \\
&= \min_{\gamma \in \mathcal{R}}{\{\gamma + \frac{1}{1+\alpha}E[{(t_{ij}^D-\gamma)}^+]\}}
\end{align}
\par
Now we turn to the queue at the ECS. Similar to (\ref{approx_CDF_W_D}), the CDF of $W_{ij}^S$ can be approximated as
\begin{equation}
\label{approx_W_S}
F(W_{ij}^D) \approx 1 - {\rm exp}\left[{- \frac{2(1 - \rho_{ij}^S)}{\lambda_{ij}^S {\sigma^S_{ij}}^2} W_{ij}^S}\right],
\end{equation}
and thus the CVaR of the queuing delay at the ECS can be evaluated in the following theorem.
\begin{theorem}
\label{Theorem3}
For a confidence level $\alpha \in (0,1)$, the $\alpha$-CVaR of $W_{ij}^S$ can be expressed as
\begin{equation}
\label{CVaR_W_S}
{\rm CVaR}_\alpha(W_{ij}^S) = \frac{\lambda_{ij}^S {\sigma^S_{ij}}^2}{2(1- \rho_{ij}^S)} [1 - \ln{(1-\alpha)}].
\end{equation}
\end{theorem}
The proof of Theorem \ref{Theorem3} is similar to Theorem \ref{Theorem2} and is omitted here for brevity.
\par
Since we assume constant computing capability at the ECS, the CVaR of the service time of $Q_{ij}^S$  is at the same value as itself, i.e.
\begin{equation}
\label{CVaR_t_S}
{\rm CVaR}_\alpha(t_{ij}^S) = \frac{c_i}{f_{ij}}.
\end{equation}
\par
With the CVaR of each part of the delay involved in the task offloading, we provide an upper bound of the CVaR of the total delay in the following theorem based on the convexity and sub-additivity \cite{pflug2000some}.
\begin{theorem}
\label{Theorem4}
For a confidence level $\alpha \in (0,1)$, an upper bound of the $\alpha$-CVaR of $t_i$ is  given by
\begin{equation}
\label{CVaR_t}
{\rm CVaR}_\alpha(t_i) \leq {\rm CVaR}_\alpha(t_i)^\ast,
\end{equation}
where
\begin{equation}
{\rm CVaR}_\alpha(t_i)^\ast =  \sum_{j=1}^N x_{ij} \left[  {\rm CVaR}_\alpha(W_{ij}^D)  + {\rm CVaR}_\alpha(t_{ij}^D) + {\rm CVaR}_\alpha(W_{ij}^S) + {\rm CVaR}_\alpha(t_{ij}^S) \right].
\end{equation}
\end{theorem}
\begin{proof}
\label{proofofTheorem4}
According to the convexity, the $\alpha$-CVaR of $t_i$ satisfies the following Jensen inequality:
\begin{equation}
\label{JensenInequality}
{\rm CVaR}_\alpha(t_i) 
= {\rm CVaR}_\alpha(\sum_{j=1}^N{x_{ij}t_{ij}}) \leq \sum_{j=1}^N{x_{ij} {\rm CVaR}_\alpha(t_{ij})}.
\end{equation}
Furthermore, based on (\ref{totaldelay}) and the sub-additivity of the CVaR, we have the following inequality:
\begin{equation}
\label{SubaddInequality}
{\rm CVaR}_\alpha(t_{ij}) \leq {\rm CVaR}_\alpha(W_{ij}^D) + {\rm CVaR}_\alpha(t_{ij}^D) + {\rm CVaR}_\alpha(W_{ij}^S) + {\rm CVaR}_\alpha(t_{ij}^S).
\end{equation}
By combining (\ref{JensenInequality}) and (\ref{SubaddInequality}), Theorem \ref{Theorem4} is proved.
\end{proof}

\section{Problem formulation and solution}
\subsection{Problem formulation}
In the design of the edge computing-assisted IIoT system, not only the average latency but also risk behind the intense delay jitter should be carefully considered. Taking into account both the average delay performance and the risk, we set the objective of the task offloading problem as the weighted sum of the average delay and the CVaR, i.e. the mean-risk sum. We have shown that obtaining an explicit expression of both the two terms is often cumbersome, especially for the complex wireless environment in the automated factories. Therefore, the two upper bounds of the average total delay and the corresponding CVaR derived in the previous section are adopted instead. Furthermore, in the considered mission-critical IIoT scenario, the performance of the whole system is usually determined by the device with the worst performance. As a result, we aim to minimize the maximum mean-risk sum among all the devices, which can be described as the following optimization problem:
\begin{subequations}
\label{P1}
\begin{alignat}{2}
\min_{\mathbf{X},\mathbf{f}} \max_{i\in\mathcal{M}} \quad & E[t_i]^\ast + \beta {\rm CVaR}_\alpha(t_i)^\ast &\quad &\tag{28}\\
\text{s.t.}\quad 
&x_{ij}\in\{0,1\},\forall i\in\mathcal{M},\forall j\in\mathcal{N}, \label{P1:C1}\\
&\sum_{j=1}^N{x_{ij}}=1,\forall i\in\mathcal{M}, \label{P1:C2}\\
&\sum_{i=1}^M{x_{ij}}\leq N_j,\forall j\in\mathcal{N}, \label{P1:C3}\\
&\sum_{i=1}^M{f_{ij}\leq F_j},\forall j \in \mathcal{N}, \label{P1:C4}
\end{alignat}
\end{subequations}
where $\mathbf{f}=[f_{ij}],i\in\mathcal{M},j\in\mathcal{N}$ is the computation frequency allocation matrix, $\beta\in(0,1)$ is the weight of the CVaR, also called the risk-sensitive parameter, and $F_j$ is the overall computation frequency of ECS $j$. Constraint (\ref{P1:C2}) is used to guarantee that the tasks generated by the device can only be offloaded to one ECS. Constraint (\ref{P1:C3}) and (\ref{P1:C4}) indicate that the number of devices served by a ECS should not exceed the number of its CPU cores, and the sum of the computation frequency allocated to these devices should not exceed its overall computation frequency. Substituting (\ref{E_t_ast}), (\ref{CVaR_W_D_cal}), (\ref{CVaR_t_D_cal}), (\ref{CVaR_W_S}) and (\ref{CVaR_t_S}) to (\ref{P1}), we find that the optimization problem is a non-convex mixed integer non-linear problem (MINLP), which is NP-hard \cite{bertsimas1997introduction}. To reduce the computation overhead, we propose a heuristic algorithm, which will be described in details in the following subsection.
\subsection{Problem solving}
Recall that the CVaR of the queuing delay at the device in (\ref{CVaR_t_D_cal}) is in the form of an minimization problem. Since the optimization variable $\gamma$ in (\ref{CVaR_t_D_cal}) is independent of the optimization variables $\mathbf{X}$ and $\mathbf{f}$ in (\ref{P1}), we can solve (\ref{CVaR_t_D_cal}) first and substitute its optimal solution to (\ref{P1}) for the subsequent problem solving.
\par
To solve (\ref{CVaR_t_D_cal}), we introduce an auxiliary variable $z_ij=(t_{ij}^D-\gamma)^+$, and problem (\ref{CVaR_t_D_cal}) can be transformed to the following problem
\begin{subequations}
\label{P2}
\begin{alignat}{2}
\min_{\gamma\in\mathcal{R},z_{ij}} \quad & \gamma + \frac{1}{1+\alpha}E[z_{ij}] &\quad & \tag{29}\\
\text{s.t.}\quad 
& z_{ij} \geq t_{ij}^D - \gamma, \label{P2:C1}\\
& z_{ij} \geq 0. \label{P2:C2}
\end{alignat}
\end{subequations}
Problem (\ref{P2}) is a stochastic optimization problem with the expectation taken over the channel gain $g_{ij}$. To approximate the expectation, we sample the probability distribution of $g_{ij}$ \cite{rockafellar2000optimization}, and a transformed problem is obtained as follows:
\begin{subequations}
\label{P3}
\begin{alignat}{2}
\min_{\gamma\in\mathcal{R},\mathbf{z_{ij}}} \quad & \gamma + \frac{1}{1+\alpha} \frac{1}{K} \sum_{k=1}^K z_{ij}^k &\quad & \tag{30}\\
\text{s.t.}\quad 
& z_{ij}^k \geq t_{ij}^D - \gamma, \forall k\in \mathcal{K}, \label{P3:C1}\\
& z_{ij}^k \geq 0, \forall k\in \mathcal{K}, \label{P3:C2}
\end{alignat}
\end{subequations}
where $z_{ij}^k, k \in \mathcal{K}= \{1,2,\cdots,K\}$ are the samples of $z_{ij}$. Problem (\ref{P3}) is a linear optimization problem, the optimal solution of which, denoted by $U_{ij}$, can be obtained from the interior point method (IPM) \cite{boyd2004convex}. There are $K+1$ variables and $2K$ constraints in problem (\ref{P3}), so we can solve it at the time complexity of $O(((3K+1)(K+1)^2+(3K+1)^{1.5}(K+1))\delta)$, where $\delta$ is the number of decoded bits \cite{vaidya1990algorithm}.
\par
With all the derived average and CVaR terms, problem (\ref{P1}) can be reformulated as the following optimization problem:
\begin{subequations}
\label{P4}
\begin{alignat}{2}
\min_{\mathbf{X},\mathbf{f}} \max_{i\in\mathcal{M}} \quad & 
\sum_{j=1}^N x_{ij} 
\Biggl\{ 
\biggl[ 
\frac{\lambda_i}{2\mu^2_{ij}(1-\rho_{ij})}
+ \frac{1}{\mu_{ij}}
+ \frac{\lambda^S_{ij}{\sigma^S_{ij}}^2}{2(1-\rho^S_{ij})} + \frac{c_i}{f_{ij}} \biggr] + 
\notag \\
& \beta 
\biggl[ 
\bigl[\frac{\lambda_{ij}^S {\sigma^S_{ij}}^2}{2(1- \rho_{ij}^S)} + \frac{\lambda_i V_{ij}}{2(1- \rho_{ij})}  \bigr] [1 - \ln{(1-\alpha)}] 
+ U_{ij} 
+ \frac{c_{i}}{f_{ij}} 
\biggr]
\Biggr\}
&\quad &\tag{31}\\
\text{s.t.}\quad 
& {\rm (28a),(28b),(28c),(28d)}. \label{P4:C1}
\end{alignat}
\end{subequations}
It is obvious that the objective function of problem (\ref{P4}) contains the term $x_{ij}/f_{ij}$, and thus problem (\ref{P4}) is still a non-convex MINLP. In order to reduce the computational complexity, we decompose the original problem into two sub-problems. First, we consider the following problem:
\begin{subequations}
  \label{P5}
  \begin{alignat}{2}
  \min_{\mathbf{X}} \max_{i\in\mathcal{M}} \quad & 
  \sum_{j=1}^N x_{ij} V_{ij}
  &\quad &\tag{32}\\
  \text{s.t.}\quad 
  & {\rm (28a),(28b),(28c)}. \label{P5:C1}
  \end{alignat}
  \end{subequations}
where
\begin{equation}
V_{ij} = 
  \frac{\lambda_i}{2\mu^2_{ij}(1-\rho_{ij})} + \frac{1}{\mu_{ij}} + \beta 
  \biggl[ 
  \frac{\lambda_i V_{ij}}{2(1- \rho_{ij})} [1 - \ln{(1-\alpha)}] 
  + U_{ij} 
  \biggr].
\end{equation}
It is worthwhile to mention that problem (\ref{P5}) is a convex MINLP, which can generally be solved via an outer approximation algorithm or an extended cutting plane algorithm \cite{kronqvist2019review}. More specifically, problem (\ref{P5}) is in the form of a bottleneck generalized assignment problem \cite{mazzola1988bottleneck}, which can be solved through the algorithm proposed in \cite{martello1995bottleneck} at the time complexity of $O(MN\log{N}+\theta(NM+N^2))$, where $\theta$ is the number of bits required to encode $\max_{i,j}{V_{ij}}$. After solving the optimal offloading matrix for problem (\ref{P5}) denoted by $X^\ast = [x_{ij}^\ast], i\in \mathcal{M},j\in \mathcal{N}$, we substitute $X^\ast$ to (\ref{P1}), and the second sub-problem can be formulated as follows:
\begin{subequations}
\label{P6}
\begin{alignat}{2}
\min_{\mathbf{f}} \max_{i\in\mathcal{M}} \quad & 
\sum_{j=1}^N x_{ij}^\ast 
\Biggl\{ 
\biggl[ 
\frac{\lambda_i}{2\mu^2_{ij}(1-\rho_{ij})}
+ \frac{1}{\mu_{ij}}
+ \frac{\lambda^S_{ij}{\sigma^S_{ij}}^2}{2(1-\rho^S_{ij})} + \frac{c_i}{f_{ij}} \biggr] + 
\notag \\
& \beta 
\biggl[ 
\bigl[\frac{\lambda_{ij}^S {\sigma^S_{ij}}^2}{2(1- \rho_{ij}^S)} + \frac{\lambda_i V_{ij}}{2(1- \rho_{ij})} \bigr] [1 - \ln{(1-\alpha)}] 
+ U_{ij} 
+ \frac{c_{i}}{f_{ij}} 
\biggr]
\Biggr\}
&\quad &\tag{34}\\
\text{s.t.}\quad 
& {\rm (28d)}. \label{P6:C1}
\end{alignat}
\end{subequations}
Problem (\ref{P6}) is a non-convex optimization problem. To transform it to a convex problem, we introduce an auxiliary variable $\mathbf{G}=[1/f_{ij}],i\in\mathcal{M},j\in\mathcal{N}$, and an optimization problem equivalent to (\ref{P6}) is given by
\begin{subequations}
\label{P7}
\begin{alignat}{2}
\min_{\mathbf{G}} \max_{i\in\mathcal{M}} \quad & 
\sum_{j=1}^N x_{ij}^\ast 
\Biggl\{ 
\biggl[ 
\frac{\lambda_i}{2\mu^2_{ij}(1-\rho_{ij})}
+ \frac{1}{\mu_{ij}}
+ \frac{\lambda^S_{ij}{\sigma^S_{ij}}^2}{2(1-\rho^S_{ij})} + c_i G_{ij} \biggr] +
\notag \\
& \beta 
\biggl[ 
\bigl[\frac{\lambda_{ij}^S {\sigma^S_{ij}}^2}{2(1- \rho_{ij}^S)} + \frac{\lambda_i V_{ij}}{2(1- \rho_{ij})} \bigr] [1 - \ln{(1-\alpha)}] 
+ U_{ij} 
+ c_i G_{ij}
\biggr]
\Biggr\}
&\quad &\tag{35}\\
\text{s.t.}\quad 
&\frac{1}{F_j} \leq G_{ij} \leq \frac{1}{x_{ij}^\ast \lambda_{ij}^S c_i}, \forall i\in \mathcal{M}, \forall j \in \mathcal{N}, \label{P7:C1} \\
&\sum_{i=1}^M{\frac{x_{ij}^\ast}{G_{ij}}} \leq F_j,\forall j \in \mathcal{N}. \label{P7:C2}
\end{alignat}
\end{subequations}
The right side of inequality (\ref{P7:C1}) is used to maintain the stability of the queue at the ECS. Although problem (\ref{P7}) is a convex optimization problem, the objective function is in the form of the pointwise maximum of $M$ mean-risk sums. To handle this, we optimize the epigraph of problem (\ref{P7}) and obtain the equivalent problem as follows:
\begin{subequations}
  \label{P8}
  \begin{alignat}{2}
  \min_{\mathbf{G},T} \quad & T \quad &\tag{36}\\
  \text{s.t.} \quad 
  &\sum_{j=1}^N x_{ij}^\ast 
  \Biggl\{ 
  \biggl[ 
  \frac{\lambda_i}{2\mu^2_{ij}(1-\rho_{ij})}
  + \frac{1}{\mu_{ij}}
  + \frac{\lambda^S_{ij}{\sigma^S_{ij}}^2}{2(1-\rho^S_{ij})} + c_i G_{ij} \biggr] +
  \notag \\
  & \beta 
  \biggl[ 
  \bigl[\frac{\lambda_{ij}^S {\sigma^S_{ij}}^2}{2(1- \rho_{ij}^S)} + \frac{\lambda_i V_{ij}}{2(1- \rho_{ij})}  \bigr] [1 - \ln{(1-\alpha)}] 
  + U_{ij} 
  + c_i G_{ij}
  \biggr]
  \Biggr\} \leq T, \forall i\in \mathcal{M},\label{P8:C1} \\
  &{\rm (35a),(35b)}. \label{P8:C2}
  \end{alignat}
\end{subequations}
\par
Problem (\ref{P8}) is a convex non-linear optimization problem, which can be solved by various algorithms such as IPM. Constraint (\ref{P8:C1}), (\ref{P7:C1}) and (\ref{P7:C2}) consists of $M$, $\sum_{i=1}^M{\sum_{j=1}^N{x_{ij}^\ast}}$ and $N$ individual constraints, respectively. As a result, problem~(\ref{P8}) has $L \triangleq M+N+\sum_{i=1}^M{\sum_{j=1}^N{x_{ij}^\ast}}$ constraints in all. According to \cite{den2012interior}, we can find an $\epsilon$-optimal solution to problem (\ref{P8}) in $O(\kappa \sqrt{L} \ln{\frac{L\mu_0}{\epsilon}})$ Newton iterations through the logarithmic barrier method \cite{den1992classical}, where $\kappa$ is the self-concordance factor, $\mu_0$ is the initial barrier value and $\epsilon$ is the accuracy parameter.
Till now, both the offloading matrix $\mathbf{X}$ and the frequency allocation matrix $\mathbf{f}$ have been solved.

\section{Results and discussion}
In this section, we evaluate the proposed strategy through numerical results. We consider a typical use case in the edge computing-assisted IIoT system, i.e. the video-operated remote control use case, with a typical latency requirement of 10-100 ms and payload size of 15-150 kbytes \cite{IIoT2019}. In the simulation, we consider 8 IIoT devices offload their tasks to 2 ECSs. Without loss of generality, we set the data size to 0.5 Mbits, i.e. 62.5 kbytes, and the task computation intensity to 15 cycles/bit for each task of each device. Each ECS is equipped with a four-core CPU. The task arrival rate of each device, i.e. the parameter of the Poisson process, is set to be uniformly distributed in $(10,30)$. The bandwidth of each wireless channel is 10 MHz, and the transmission power, noise power at the receiver and the path loss are all set to be identical for each device at 30 dBm, $10^{-9}$ W and 70 dB, respectively. To characterize the fading channel in the practical automated factory, we set the channel distribution as a mixture of Rayleigh and log-normal distribution, which has been confirmed by the measurements in the real industrial environment \cite{olofsson2016modeling}. The parameter of the Rayleigh distribution is set to be uniformly distributed in $(0.5,1)$ for each IIoT-ECS pair, and correspondingly, the two parameters of the log-normal distribution are set to be uniformly distributed in $(1,2)$ and $(0,4)$ respectively. Finally, we set the confidence level to $\alpha = 0.99$.
\par
Our proposed strategy considers both the queuing effect and the risk behind the total delay, and thus is denoted by queuing-based and risk-sensitive (Q-R) strategy in the following simulations. To evaluate the performance of the Q-R, we compare it with the following five strategies: i) the queuing-based and risk-sensitive optimal (Q-R-Opt) strategy, i.e. the globally optimal solution to problem (\ref{P1}); ii) the queuing-based and non-risk-sensitive (Q-NR) strategy, which considers the queuing effect but only optimizes the average delay performance, i.e. the weight of the CVaR is set to $\beta=0$; iii) the queuing-based and non-risk-sensitive optimal (Q-NR-Opt) strategy, i.e. the globally optimal solution that corresponds to the Q-NR case;  iv) the non-queuing-based and risk-sensitive (NQ-R) strategy, which takes into account both the average delay and the CVaR, but does not consider the queuing effect; v) the non-queuing-based and non-risk-sensitive (NQ-NR), which considers neither the queuing effect nor the risk. In the following simulations, we set the weight of the CVaR to $\beta=2$ for risk-sensitive strategies.
\par
\begin{figure}[!t]
\includegraphics[width=3.3in]{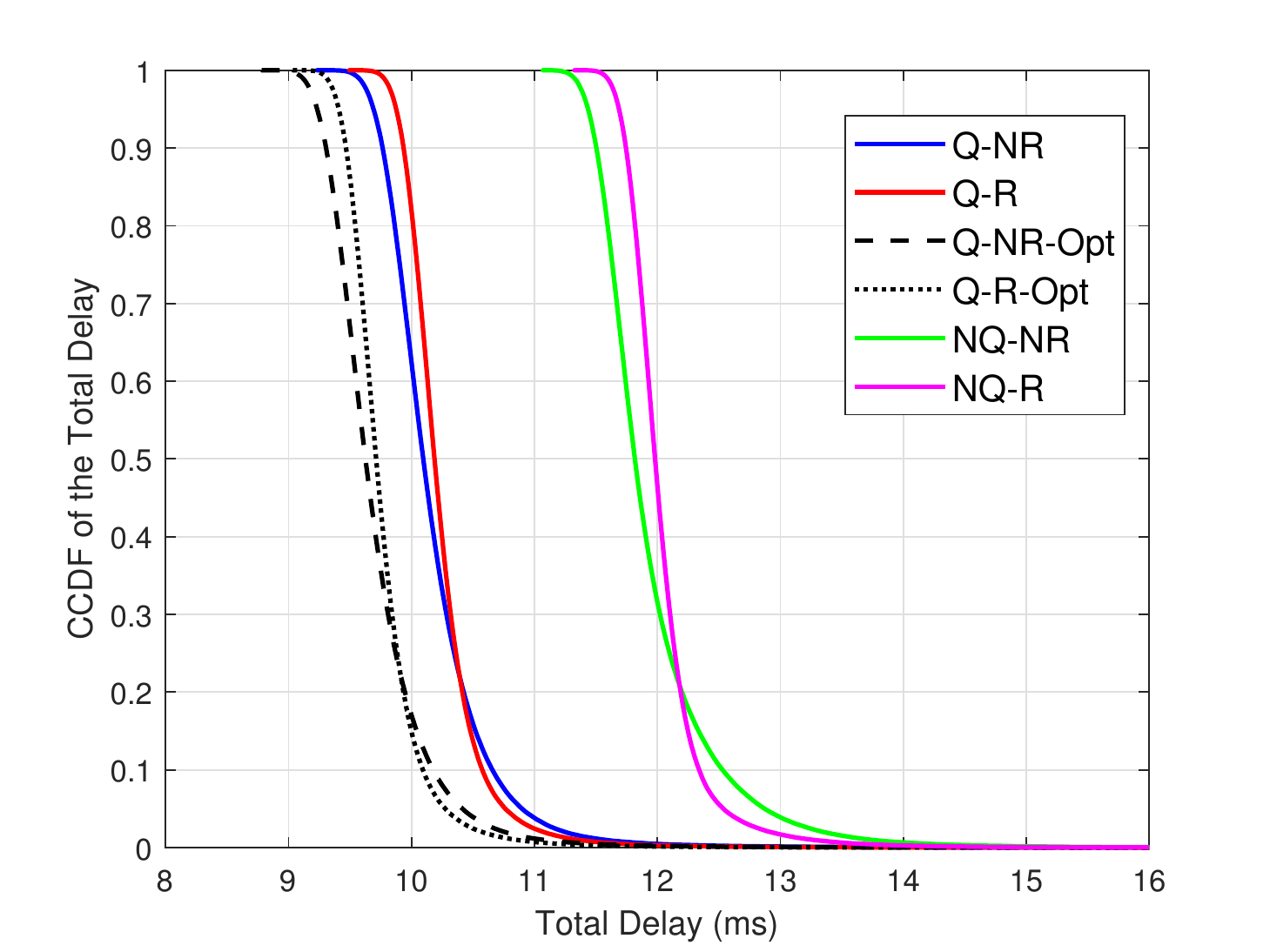}
\caption{CCDF of the total delay}
\end{figure}
We first investigate the complementary cumulative distribution function (CCDF) of the total delay under the six offloading strategies, since the CVaR captures the tail information of the delay distribution. As presented in Fig. 3, for the probability of ultra-high delay, the curve of Q-R, Q-R-Opt and NQ-R are all lower than their corresponding non-risk-sensitive strategies. 
This implies that by adding the CVaR to the optimization objective, the risk of high delay can be greatly reduced. On the other hand, we can see that for any value of the total delay, the CCDF under NQ-R and NQ-NR is greater than that under Q-R, Q-R-Opt, Q-NR and Q-NR-Opt, which means the non-queuing strategies are more likely to arise a higher delay. This is reasonable, since the non-queuing strategies neglect the queuing effect in the strategy design, which leads to a higher queuing delay. Furthermore, the curve of Q-R is close to that of Q-R-Opt and has nearly the same trend, which indicates that the proposed algorithm achieves near-optimal performance with a great reduction in computation complexity. 
\par
\begin{figure}[!t]
\includegraphics[width=3.3in]{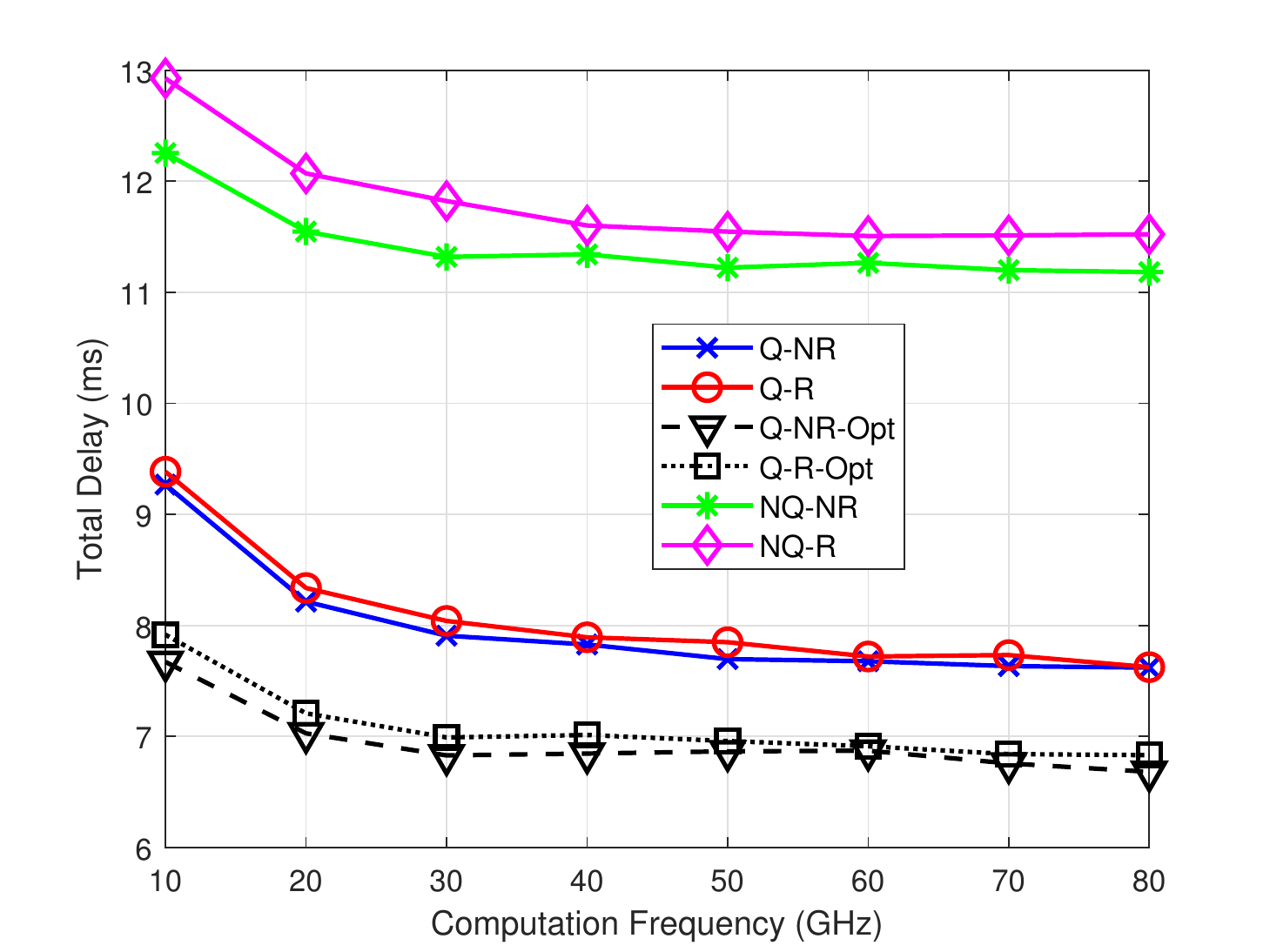}
\caption{Average total delay versus the computation frequency}
\end{figure}

\begin{figure}[!t]
\includegraphics[width=3.3in]{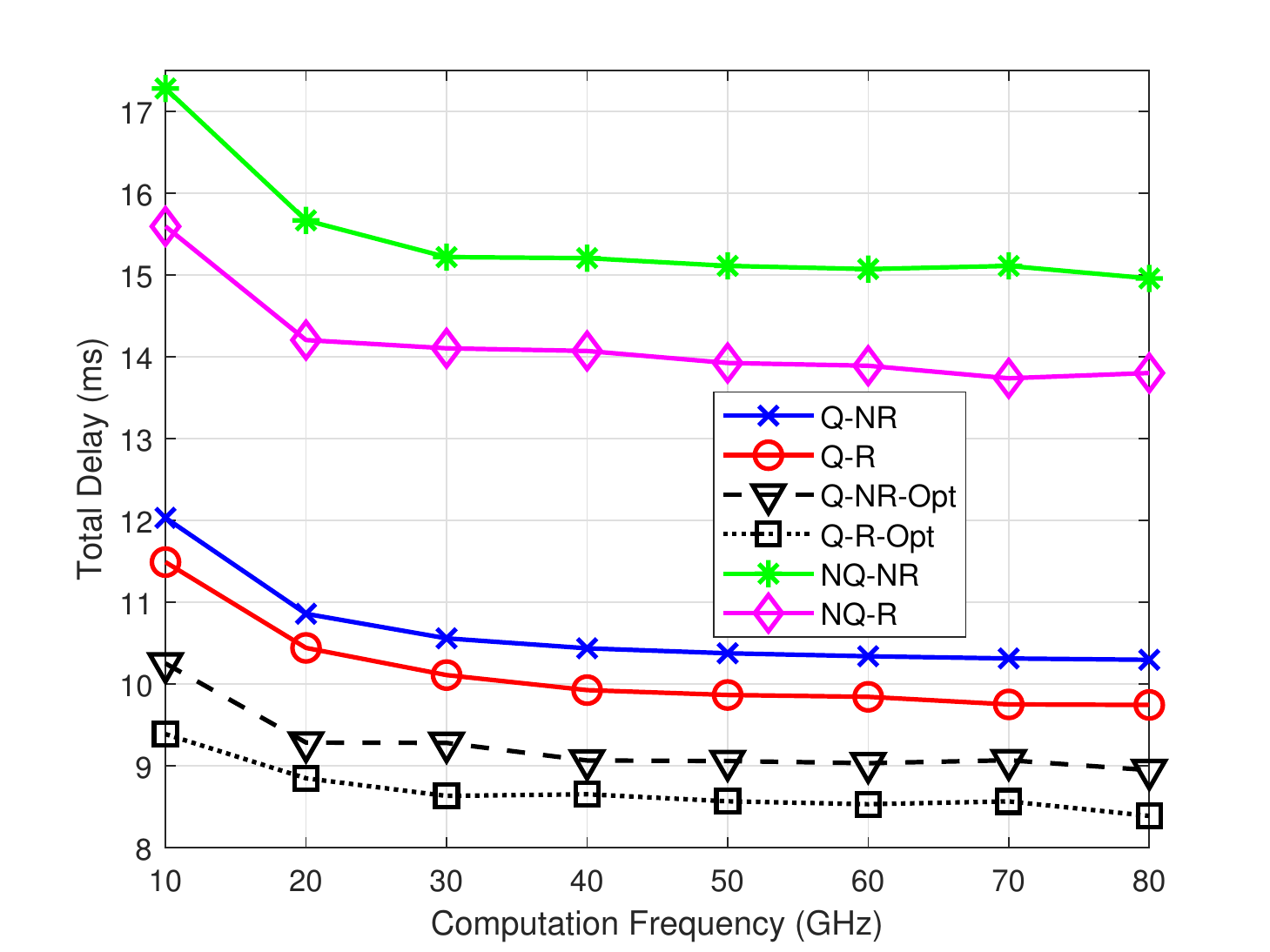}
\caption{99th percentile of the total delay versus the computation frequency}
\end{figure}
In Fig. 4 and Fig. 5, we compare how the delay performance evolves with the computation frequency of the ECS under the six strategies. Specifically, Fig.4 investigates relationship between the average delay and the computation frequency, and Fig. 5 focus on the 99th percentile of the total delay. It can be seen that with the increasing computation frequency, the average total delay and the 99th percentile decreases for all the six strategies. The reason is that the higher the computation frequency, the more computation resources allocated to the IIoT devices and thus the lower the computation delay. Furthermore, note that the delay doesn't descend much when computation frequency is relatively high. This is because for high computation frequency, both the computation delay and the queuing delay at the ECS is relatively low, and thus the total delay is mainly dependent on the queuing delay at the devices. We can also see that the queuing strategies outperforms the corresponding non-queuing strategies for both the average performance and the 99th percentile, which verifies the significance of the queuing analysis again. More importantly, the two figures verify the near-optimality of the proposed algorithm, and jointly indicate that the risk-sensitive strategies achieve nearly the same average total delay as the non-risk-sensitive ones, but greatly reduce the 99th percentile of the total delay by incorporating the risk to the design of offloading strategy. In other words, the intense delay jitter can be effectively controlled under the risk-sensitive strategies only at the price of very little degradation on the average performance.
\par
\begin{figure}[!t]
\includegraphics[width=3.3in]{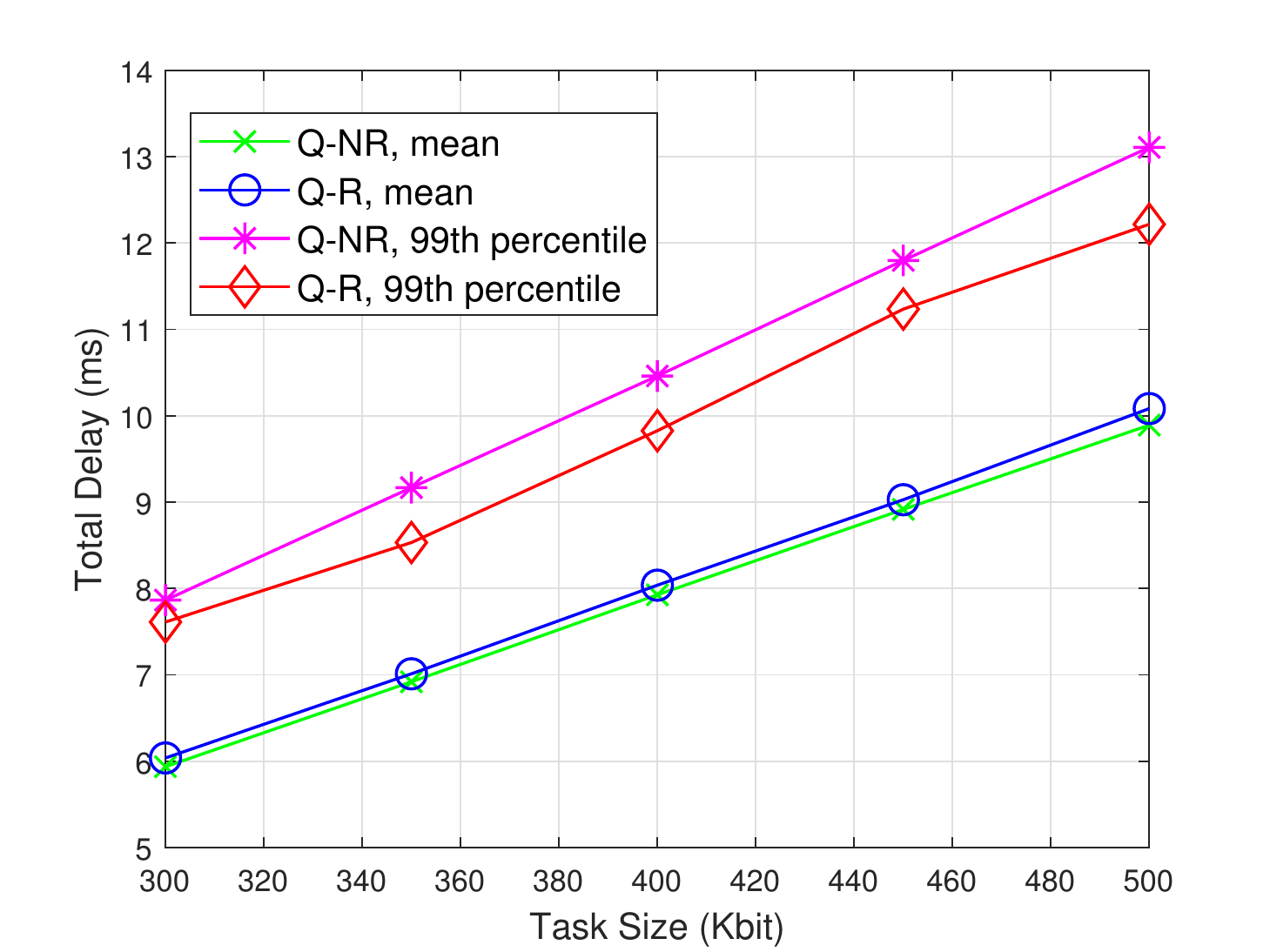}
\caption{Delay performance versus the task size where the computation frequency is 10 GHz}
\end{figure}
Finally, we investigate the effect of the task size on both the average delay and the 99th percentile under the Q-R and Q-NR strategies. As shown in Fig. 6, the 99th percentile is higher than the average total delay for both strategies, since the former characterize the worst case delay. With the increase of the task size, both the average delay and the 99th percentile increase under both strategies. This is due to the fact that a larger task size lead to the higher transmission and computation delay. Furthermore, the 99th percentile of Q-R is always lower than that of Q-NR, while the two curves of the average delay almost coincide with each other. Note that the average delay under the Q-NR strategy is the lower bound of that under the Q-R. This implies that Q-R achieves nearly the same average performance as the Q-NR while simultaneously improving the worst case performance with respect to the total delay.

\section{Conclusions}
In this work, we introduce the risk management theory to design of the edge computing-assisted IIoT system. We explore the queuing theory and the properties of the CVaR to capture the tail distribution of the end-to-end delay, and provide two upper bounds of the average total delay and the CVaR. A joint task offloading and computation resource allocation problem is formulated to simultaneously minimize the average total delay and the risk. Since the problem is a non-convex MINLP, we decompose it into two sub-problems and design a two-stage heuristic algorithm. The computation complexity of each procedure of the proposed algorithm has been analyzed. Finally, simulations are performed under the practical channel model in the automated factories, and the results verify that the proposed strategy can effectively control the risk of intense delay jitter while guaranteeing the average delay performance. 


\begin{backmatter}

\section*{Acknowledgements}
Not applicable

\section*{Funding}
This work was supported by the Shanghai Municipal Natural Science Foundation (No. 19ZR1404700). This work was also supported in part by the National Science Foundation of China under Grants 71731004.

\section*{Abbreviations}
URLLC: Ultra-reliable low-latency communications; 5G: fifth generation; 6G: sixth generation; MDP: Markov decision process; ADMM: Alternating Direction Method of Multipliers; IIoT: Industrial Internet of Things; VaR: Value at Risk; CVaR: Conditional Value at Risk; ECS: Edge computing server; CDF: Cumulative distribution function; MINLP: Mixed integer non-linear problem; IPM: Interior point method; Q-R: Queuing-based and risk-sensitive; Q-NR: Queuing-based and non-risk-sensitive; Q-R-Opt: Queuing-based and risk-sensitive optimal; Q-NR-Opt: Queuing-based and non-risk-sensitive optimal; NQ-R: Non-queuing-based and risk-sensitive; NQ-NR: Non-queuing-based and non-risk-sensitive; CCDF: Complementary cumulative distribution function

\section*{Availability of data and materials}
Not applicable

\section*{Ethics approval and consent to participate}
Not applicable

\section*{Competing interests}
The authors declare that they have no competing interests.

\section*{Consent for publication}
Not applicable

\section*{Authors' contributions}
All the authors contributed to the system design, queuing and CVaR analysis, strategy design, simulations and the writing of this paper. All authors read and approved the final manuscript.

\section*{Authors' information}
Not applicable


\bibliographystyle{bmc-mathphys} 
\bibliography{risk_sensitive_edge_computing}      


\begin{thebibliography}{37}
\ifx \bisbn   \undefined \def \bisbn  #1{ISBN #1}\fi
\ifx \binits  \undefined \def \binits#1{#1}\fi
\ifx \bauthor  \undefined \def \bauthor#1{#1}\fi
\ifx \batitle  \undefined \def \batitle#1{#1}\fi
\ifx \bjtitle  \undefined \def \bjtitle#1{#1}\fi
\ifx \bvolume  \undefined \def \bvolume#1{\textbf{#1}}\fi
\ifx \byear  \undefined \def \byear#1{#1}\fi
\ifx \bissue  \undefined \def \bissue#1{#1}\fi
\ifx \bfpage  \undefined \def \bfpage#1{#1}\fi
\ifx \blpage  \undefined \def \blpage #1{#1}\fi
\ifx \burl  \undefined \def \burl#1{\textsf{#1}}\fi
\ifx \doiurl  \undefined \def \doiurl#1{\textsf{#1}}\fi
\ifx \betal  \undefined \def \betal{\textit{et al.}}\fi
\ifx \binstitute  \undefined \def \binstitute#1{#1}\fi
\ifx \binstitutionaled  \undefined \def \binstitutionaled#1{#1}\fi
\ifx \bctitle  \undefined \def \bctitle#1{#1}\fi
\ifx \beditor  \undefined \def \beditor#1{#1}\fi
\ifx \bpublisher  \undefined \def \bpublisher#1{#1}\fi
\ifx \bbtitle  \undefined \def \bbtitle#1{#1}\fi
\ifx \bedition  \undefined \def \bedition#1{#1}\fi
\ifx \bseriesno  \undefined \def \bseriesno#1{#1}\fi
\ifx \blocation  \undefined \def \blocation#1{#1}\fi
\ifx \bsertitle  \undefined \def \bsertitle#1{#1}\fi
\ifx \bsnm \undefined \def \bsnm#1{#1}\fi
\ifx \bsuffix \undefined \def \bsuffix#1{#1}\fi
\ifx \bparticle \undefined \def \bparticle#1{#1}\fi
\ifx \barticle \undefined \def \barticle#1{#1}\fi
\ifx \bconfdate \undefined \def \bconfdate #1{#1}\fi
\ifx \botherref \undefined \def \botherref #1{#1}\fi
\ifx \url \undefined \def \url#1{\textsf{#1}}\fi
\ifx \bchapter \undefined \def \bchapter#1{#1}\fi
\ifx \bbook \undefined \def \bbook#1{#1}\fi
\ifx \bcomment \undefined \def \bcomment#1{#1}\fi
\ifx \oauthor \undefined \def \oauthor#1{#1}\fi
\ifx \citeauthoryear \undefined \def \citeauthoryear#1{#1}\fi
\ifx \endbibitem  \undefined \def \endbibitem {}\fi
\ifx \bconflocation  \undefined \def \bconflocation#1{#1}\fi
\ifx \arxivurl  \undefined \def \arxivurl#1{\textsf{#1}}\fi
\csname PreBibitemsHook\endcsname

\bibitem{gu2020knowledge}
\begin{botherref}
\oauthor{\bsnm{Gu}, \binits{Z.}},
\oauthor{\bsnm{She}, \binits{C.}},
\oauthor{\bsnm{Hardjawana}, \binits{W.}},
\oauthor{\bsnm{Lumb}, \binits{S.}},
\oauthor{\bsnm{McKechnie}, \binits{D.}},
\oauthor{\bsnm{Essery}, \binits{T.}},
\oauthor{\bsnm{Vucetic}, \binits{B.}}:
Knowledge-assisted deep reinforcement learning in 5g scheduler design: From
  theoretical framework to implementation.
arXiv preprint arXiv:2009.08346
(2020)
\end{botherref}
\endbibitem

\bibitem{she2020tutorial}
\begin{botherref}
\oauthor{\bsnm{She}, \binits{C.}},
\oauthor{\bsnm{Sun}, \binits{C.}},
\oauthor{\bsnm{Gu}, \binits{Z.}},
\oauthor{\bsnm{Li}, \binits{Y.}},
\oauthor{\bsnm{Yang}, \binits{C.}},
\oauthor{\bsnm{Poor}, \binits{H.V.}},
\oauthor{\bsnm{Vucetic}, \binits{B.}}:
A tutorial of ultra-reliable and low-latency communications in 6g: Integrating
  theoretical knowledge into deep learning.
arXiv preprint arXiv:2009.06010
(2020)
\end{botherref}
\endbibitem

\bibitem{2018SmartFactory}
\begin{barticle}
\bauthor{\bsnm{{Wan}}, \binits{J.}},
\bauthor{\bsnm{{Yang}}, \binits{J.}},
\bauthor{\bsnm{{Wang}}, \binits{Z.}},
\bauthor{\bsnm{{Hua}}, \binits{Q.}}:
\batitle{Artificial intelligence for cloud-assisted smart factory}.
\bjtitle{IEEE Access}
\bvolume{6},
\bfpage{55419}--\blpage{55430}
(\byear{2018})
\end{barticle}
\endbibitem

\bibitem{2019EC}
\begin{barticle}
\bauthor{\bsnm{{Elbamby}}, \binits{M.S.}},
\bauthor{\bsnm{{Perfecto}}, \binits{C.}},
\bauthor{\bsnm{{Liu}}, \binits{C.}},
\bauthor{\bsnm{{Park}}, \binits{J.}},
\bauthor{\bsnm{{Samarakoon}}, \binits{S.}},
\bauthor{\bsnm{{Chen}}, \binits{X.}},
\bauthor{\bsnm{{Bennis}}, \binits{M.}}:
\batitle{Wireless edge computing with latency and reliability guarantees}.
\bjtitle{Proceedings of the IEEE}
\bvolume{107}(\bissue{8}),
\bfpage{1717}--\blpage{1737}
(\byear{2019})
\end{barticle}
\endbibitem

\bibitem{velte2009cloud}
\begin{bbook}
\bauthor{\bsnm{Velte}, \binits{T.}},
\bauthor{\bsnm{Velte}, \binits{A.}},
\bauthor{\bsnm{Elsenpeter}, \binits{R.}}:
\bbtitle{Cloud Computing, a Practical Approach}.
\bpublisher{McGraw-Hill, Inc.},
\blocation{USA}
(\byear{2009})
\end{bbook}
\endbibitem

\bibitem{2016EC_VisionandChallenge}
\begin{barticle}
\bauthor{\bsnm{{Shi}}, \binits{W.}},
\bauthor{\bsnm{{Cao}}, \binits{J.}},
\bauthor{\bsnm{{Zhang}}, \binits{Q.}},
\bauthor{\bsnm{{Li}}, \binits{Y.}},
\bauthor{\bsnm{{Xu}}, \binits{L.}}:
\batitle{Edge computing: Vision and challenges}.
\bjtitle{IEEE Internet of Things Journal}
\bvolume{3}(\bissue{5}),
\bfpage{637}--\blpage{646}
(\byear{2016}).
doi:\doiurl{10.1109/JIOT.2016.2579198}
\end{barticle}
\endbibitem

\bibitem{2016MEC}
\begin{bchapter}
\bauthor{\bsnm{{Liu}}, \binits{J.}},
\bauthor{\bsnm{{Mao}}, \binits{Y.}},
\bauthor{\bsnm{{Zhang}}, \binits{J.}},
\bauthor{\bsnm{{Letaief}}, \binits{K.B.}}:
\bctitle{Delay-optimal computation task scheduling for mobile-edge computing
  systems}.
In: \bbtitle{2016 IEEE International Symposium on Information Theory (ISIT)},
pp. \bfpage{1451}--\blpage{1455}
(\byear{2016})
\end{bchapter}
\endbibitem

\bibitem{2019CellularEC}
\begin{barticle}
\bauthor{\bsnm{{Zhang}}, \binits{Y.}},
\bauthor{\bsnm{{Du}}, \binits{P.}},
\bauthor{\bsnm{{Wang}}, \binits{J.}},
\bauthor{\bsnm{{Ba}}, \binits{T.}},
\bauthor{\bsnm{{Ding}}, \binits{R.}},
\bauthor{\bsnm{{Xin}}, \binits{N.}}:
\batitle{Resource scheduling for delay minimization in multi-server cellular
  edge computing systems}.
\bjtitle{IEEE Access}
\bvolume{7},
\bfpage{86265}--\blpage{86273}
(\byear{2019})
\end{barticle}
\endbibitem

\bibitem{2019CooperativeEC}
\begin{barticle}
\bauthor{\bsnm{{Wang}}, \binits{Y.}},
\bauthor{\bsnm{{Tao}}, \binits{X.}},
\bauthor{\bsnm{{Zhang}}, \binits{X.}},
\bauthor{\bsnm{{Zhang}}, \binits{P.}},
\bauthor{\bsnm{{Hou}}, \binits{Y.T.}}:
\batitle{Cooperative task offloading in three-tier mobile computing networks:
  An admm framework}.
\bjtitle{IEEE Transactions on Vehicular Technology}
\bvolume{68}(\bissue{3}),
\bfpage{2763}--\blpage{2776}
(\byear{2019}).
doi:\doiurl{10.1109/TVT.2019.2892176}
\end{barticle}
\endbibitem

\bibitem{2019Adaptive}
\begin{barticle}
\bauthor{\bsnm{{Sun}}, \binits{Y.}},
\bauthor{\bsnm{{Guo}}, \binits{X.}},
\bauthor{\bsnm{{Song}}, \binits{J.}},
\bauthor{\bsnm{{Zhou}}, \binits{S.}},
\bauthor{\bsnm{{Jiang}}, \binits{Z.}},
\bauthor{\bsnm{{Liu}}, \binits{X.}},
\bauthor{\bsnm{{Niu}}, \binits{Z.}}:
\batitle{Adaptive learning-based task offloading for vehicular edge computing
  systems}.
\bjtitle{IEEE Transactions on Vehicular Technology}
\bvolume{68}(\bissue{4}),
\bfpage{3061}--\blpage{3074}
(\byear{2019})
\end{barticle}
\endbibitem

\bibitem{2018cloudfogcomputing}
\begin{bchapter}
\bauthor{\bsnm{{Shi}}, \binits{C.}},
\bauthor{\bsnm{{Ren}}, \binits{Z.}},
\bauthor{\bsnm{{Yang}}, \binits{K.}},
\bauthor{\bsnm{{Chen}}, \binits{C.}},
\bauthor{\bsnm{{Zhang}}, \binits{H.}},
\bauthor{\bsnm{{Xiao}}, \binits{Y.}},
\bauthor{\bsnm{{Hou}}, \binits{X.}}:
\bctitle{Ultra-low latency cloud-fog computing for industrial internet of
  things}.
In: \bbtitle{2018 IEEE Wireless Communications and Networking Conference
  (WCNC)},
pp. \bfpage{1}--\blpage{6}
(\byear{2018})
\end{bchapter}
\endbibitem

\bibitem{2019Dynamic}
\begin{barticle}
\bauthor{\bsnm{{Liu}}, \binits{C.}},
\bauthor{\bsnm{{Bennis}}, \binits{M.}},
\bauthor{\bsnm{{Debbah}}, \binits{M.}},
\bauthor{\bsnm{{Poor}}, \binits{H.V.}}:
\batitle{Dynamic task offloading and resource allocation for ultra-reliable
  low-latency edge computing}.
\bjtitle{IEEE Transactions on Communications}
\bvolume{67}(\bissue{6}),
\bfpage{4132}--\blpage{4150}
(\byear{2019})
\end{barticle}
\endbibitem

\bibitem{2019FBL}
\begin{bchapter}
\bauthor{\bsnm{{Zhu}}, \binits{Y.}},
\bauthor{\bsnm{{Hu}}, \binits{Y.}},
\bauthor{\bsnm{{Yang}}, \binits{T.}},
\bauthor{\bsnm{{Schmeink}}, \binits{A.}}:
\bctitle{Reliability-optimal offloading in multi-server edge computing networks
  with transmissions carried by finite blocklength codes}.
In: \bbtitle{2019 IEEE International Conference on Communications Workshops
  (ICC Workshops)},
pp. \bfpage{1}--\blpage{6}
(\byear{2019}).
doi:\doiurl{10.1109/ICCW.2019.8757175}
\end{bchapter}
\endbibitem

\bibitem{2020RiskSensitiveVEC}
\begin{barticle}
\bauthor{\bsnm{{Batewela}}, \binits{S.}},
\bauthor{\bsnm{{Liu}}, \binits{C.}},
\bauthor{\bsnm{{Bennis}}, \binits{M.}},
\bauthor{\bsnm{{Suraweera}}, \binits{H.A.}},
\bauthor{\bsnm{{Hong}}, \binits{C.S.}}:
\batitle{Risk-sensitive task fetching and offloading for vehicular edge
  computing}.
\bjtitle{IEEE Communications Letters}
\bvolume{24}(\bissue{3}),
\bfpage{617}--\blpage{621}
(\byear{2020}).
doi:\doiurl{10.1109/LCOMM.2019.2960777}
\end{barticle}
\endbibitem

\bibitem{de2007extreme}
\begin{bbook}
\bauthor{\bsnm{De~Haan}, \binits{L.}},
\bauthor{\bsnm{Ferreira}, \binits{A.}}:
\bbtitle{Extreme Value Theory: an Introduction}.
\bpublisher{Springer},
\blocation{New York}
(\byear{2007})
\end{bbook}
\endbibitem

\bibitem{mihatsch2002risk}
\begin{barticle}
\bauthor{\bsnm{Mihatsch}, \binits{O.}},
\bauthor{\bsnm{Neuneier}, \binits{R.}}:
\batitle{Risk-sensitive reinforcement learning}.
\bjtitle{Machine learning}
\bvolume{49}(\bissue{2-3}),
\bfpage{267}--\blpage{290}
(\byear{2002})
\end{barticle}
\endbibitem

\bibitem{RiskManagement}
\begin{bbook}
\beditor{\bsnm{Hessami}, \binits{A.G.}} (ed.):
\bbtitle{{Perspectives on Risk, Assessment and Management Paradigms}}.
\bsertitle{Books},
vol. \bseriesno{5615}.
\bpublisher{IntechOpen},
\blocation{London}
(\byear{2019}).
doi:\doiurl{10.5772/intechopen.77127}.
\burl{https://ideas.repec.org/b/ito/pbooks/5615.html}
\end{bbook}
\endbibitem

\bibitem{2017EEMEC}
\begin{barticle}
\bauthor{\bsnm{{You}}, \binits{C.}},
\bauthor{\bsnm{{Huang}}, \binits{K.}},
\bauthor{\bsnm{{Chae}}, \binits{H.}},
\bauthor{\bsnm{{Kim}}, \binits{B.}}:
\batitle{Energy-efficient resource allocation for mobile-edge computation
  offloading}.
\bjtitle{IEEE Transactions on Wireless Communications}
\bvolume{16}(\bissue{3}),
\bfpage{1397}--\blpage{1411}
(\byear{2017}).
doi:\doiurl{10.1109/TWC.2016.2633522}
\end{barticle}
\endbibitem

\bibitem{bertsekas1992data}
\begin{bbook}
\bauthor{\bsnm{Bertsekas}, \binits{D.P.}},
\bauthor{\bsnm{Gallager}, \binits{R.G.}},
\bauthor{\bsnm{Humblet}, \binits{P.}}:
\bbtitle{Data Networks}
vol. \bseriesno{2}.
\bpublisher{Prentice-Hall International},
\blocation{New Jersey}
(\byear{1992})
\end{bbook}
\endbibitem

\bibitem{kleinrock1976queueing}
\begin{bbook}
\bauthor{\bsnm{Kleinrock}, \binits{L.}}:
\bbtitle{Queueing Systems, Volume 2: Computer Applications}
vol. \bseriesno{66}.
\bpublisher{wiley},
\blocation{New York}
(\byear{1976})
\end{bbook}
\endbibitem

\bibitem{bertsimas1990departure}
\begin{botherref}
\oauthor{\bsnm{Bertsimas}, \binits{D.J.}},
\oauthor{\bsnm{Nakazato}, \binits{D.}}:
The departure process from a gi/g/1 queue and its applications to the analysis
  of tandem queues
(1990)
\end{botherref}
\endbibitem

\bibitem{yeh2000characterizing}
\begin{barticle}
\bauthor{\bsnm{Yeh}, \binits{P.-C.}},
\bauthor{\bsnm{Chang}, \binits{J.-F.}}:
\batitle{Characterizing the departure process of a single server queue from the
  embedded markov renewal process at departures}.
\bjtitle{Queueing systems}
\bvolume{35}(\bissue{1-4}),
\bfpage{381}--\blpage{395}
(\byear{2000})
\end{barticle}
\endbibitem

\bibitem{jorion2000value}
\begin{botherref}
\oauthor{\bsnm{Jorion}, \binits{P.}}:
Value at risk
(2000)
\end{botherref}
\endbibitem

\bibitem{acerbi2002coherence}
\begin{barticle}
\bauthor{\bsnm{Acerbi}, \binits{C.}},
\bauthor{\bsnm{Tasche}, \binits{D.}}:
\batitle{On the coherence of expected shortfall}.
\bjtitle{Journal of Banking \& Finance}
\bvolume{26}(\bissue{7}),
\bfpage{1487}--\blpage{1503}
(\byear{2002})
\end{barticle}
\endbibitem

\bibitem{kingman1962queues}
\begin{barticle}
\bauthor{\bsnm{Kingman}, \binits{J.F.}}:
\batitle{On queues in heavy traffic}.
\bjtitle{Journal of the Royal Statistical Society: Series B (Methodological)}
\bvolume{24}(\bissue{2}),
\bfpage{383}--\blpage{392}
(\byear{1962})
\end{barticle}
\endbibitem

\bibitem{rockafellar2000optimization}
\begin{barticle}
\bauthor{\bsnm{Rockafellar}, \binits{R.T.}},
\bauthor{\bsnm{Uryasev}, \binits{S.}}, \betal:
\batitle{Optimization of conditional value-at-risk}.
\bjtitle{Journal of risk}
\bvolume{2},
\bfpage{21}--\blpage{42}
(\byear{2000})
\end{barticle}
\endbibitem

\bibitem{pflug2000some}
\begin{bchapter}
\bauthor{\bsnm{Pflug}, \binits{G.C.}}:
\bctitle{Some remarks on the value-at-risk and the conditional value-at-risk}.
In: \bbtitle{Probabilistic Constrained Optimization},
pp. \bfpage{272}--\blpage{281}.
\bpublisher{Springer},
\blocation{Boston, MA}
(\byear{2000})
\end{bchapter}
\endbibitem

\bibitem{bertsimas1997introduction}
\begin{bbook}
\bauthor{\bsnm{Bertsimas}, \binits{D.}},
\bauthor{\bsnm{Tsitsiklis}, \binits{J.N.}}:
\bbtitle{Introduction to Linear Optimization}
vol. \bseriesno{6}.
\bpublisher{Athena Scientific},
\blocation{Belmont, MA}
(\byear{1997})
\end{bbook}
\endbibitem

\bibitem{boyd2004convex}
\begin{bbook}
\bauthor{\bsnm{Boyd}, \binits{S.}},
\bauthor{\bsnm{Boyd}, \binits{S.P.}},
\bauthor{\bsnm{Vandenberghe}, \binits{L.}}:
\bbtitle{Convex Optimization}.
\bpublisher{Cambridge university press},
\blocation{UK}
(\byear{2004})
\end{bbook}
\endbibitem

\bibitem{vaidya1990algorithm}
\begin{barticle}
\bauthor{\bsnm{Vaidya}, \binits{P.M.}}:
\batitle{An algorithm for linear programming which requires o (((m+ n) n 2+(m+
  n) 1.5 n) l) arithmetic operations}.
\bjtitle{Mathematical Programming}
\bvolume{47}(\bissue{1-3}),
\bfpage{175}--\blpage{201}
(\byear{1990})
\end{barticle}
\endbibitem

\bibitem{kronqvist2019review}
\begin{barticle}
\bauthor{\bsnm{Kronqvist}, \binits{J.}},
\bauthor{\bsnm{Bernal}, \binits{D.E.}},
\bauthor{\bsnm{Lundell}, \binits{A.}},
\bauthor{\bsnm{Grossmann}, \binits{I.E.}}:
\batitle{A review and comparison of solvers for convex minlp}.
\bjtitle{Optimization and Engineering}
\bvolume{20}(\bissue{2}),
\bfpage{397}--\blpage{455}
(\byear{2019})
\end{barticle}
\endbibitem

\bibitem{mazzola1988bottleneck}
\begin{barticle}
\bauthor{\bsnm{Mazzola}, \binits{J.}},
\bauthor{\bsnm{Neebe}, \binits{A.}}:
\batitle{Bottleneck generalized assignment problems}.
\bjtitle{Engineering Costs and Production Economics}
\bvolume{14}(\bissue{1}),
\bfpage{61}--\blpage{65}
(\byear{1988})
\end{barticle}
\endbibitem

\bibitem{martello1995bottleneck}
\begin{barticle}
\bauthor{\bsnm{Martello}, \binits{S.}},
\bauthor{\bsnm{Toth}, \binits{P.}}:
\batitle{The bottleneck generalized assignment problem}.
\bjtitle{European journal of operational research}
\bvolume{83}(\bissue{3}),
\bfpage{621}--\blpage{638}
(\byear{1995})
\end{barticle}
\endbibitem

\bibitem{den2012interior}
\begin{bbook}
\bauthor{\bsnm{Den~Hertog}, \binits{D.}}:
\bbtitle{Interior Point Approach to Linear, Quadratic and Convex Programming:
  Algorithms and Complexity}
vol. \bseriesno{277}.
\bpublisher{Springer},
\blocation{Dordrecht}
(\byear{2012})
\end{bbook}
\endbibitem

\bibitem{den1992classical}
\begin{barticle}
\bauthor{\bsnm{Den~Hertog}, \binits{D.}},
\bauthor{\bsnm{Roos}, \binits{C.}},
\bauthor{\bsnm{Terlaky}, \binits{T.}}:
\batitle{On the classical logarithmic barrier function method for a class of
  smooth convex programming problems}.
\bjtitle{Journal of Optimization Theory and Applications}
\bvolume{73}(\bissue{1}),
\bfpage{1}--\blpage{25}
(\byear{1992})
\end{barticle}
\endbibitem

\bibitem{IIoT2019}
\begin{barticle}
\bauthor{\bsnm{{Yang}}, \binits{J.}},
\bauthor{\bsnm{{Ai}}, \binits{B.}},
\bauthor{\bsnm{{You}}, \binits{I.}},
\bauthor{\bsnm{{Imran}}, \binits{M.}},
\bauthor{\bsnm{{Wang}}, \binits{L.}},
\bauthor{\bsnm{{Guan}}, \binits{K.}},
\bauthor{\bsnm{{He}}, \binits{D.}},
\bauthor{\bsnm{{Zhong}}, \binits{Z.}},
\bauthor{\bsnm{{Keusgen}}, \binits{W.}}:
\batitle{Ultra-reliable communications for industrial internet of things:
  Design considerations and channel modeling}.
\bjtitle{IEEE Network}
\bvolume{33}(\bissue{4}),
\bfpage{104}--\blpage{111}
(\byear{2019}).
doi:\doiurl{10.1109/MNET.2019.1800455}
\end{barticle}
\endbibitem

\bibitem{olofsson2016modeling}
\begin{barticle}
\bauthor{\bsnm{Olofsson}, \binits{T.}},
\bauthor{\bsnm{Ahlen}, \binits{A.}},
\bauthor{\bsnm{Gidlund}, \binits{M.}}:
\batitle{Modeling of the fading statistics of wireless sensor network channels
  in industrial environments}.
\bjtitle{IEEE Transactions on Signal Processing}
\bvolume{64}(\bissue{12}),
\bfpage{3021}--\blpage{3034}
(\byear{2016})
\end{barticle}
\endbibitem

\end{thebibliography}

\newcommand{\BMCxmlcomment}[1]{}

\BMCxmlcomment{

<refgrp>

<bibl id="B1">
  <title><p>Knowledge-Assisted Deep Reinforcement Learning in 5G Scheduler
  Design: From Theoretical Framework to Implementation</p></title>
  <aug>
    <au><snm>Gu</snm><fnm>Z</fnm></au>
    <au><snm>She</snm><fnm>C</fnm></au>
    <au><snm>Hardjawana</snm><fnm>W</fnm></au>
    <au><snm>Lumb</snm><fnm>S</fnm></au>
    <au><snm>McKechnie</snm><fnm>D</fnm></au>
    <au><snm>Essery</snm><fnm>T</fnm></au>
    <au><snm>Vucetic</snm><fnm>B</fnm></au>
  </aug>
  <source>arXiv preprint arXiv:2009.08346</source>
  <pubdate>2020</pubdate>
</bibl>

<bibl id="B2">
  <title><p>A Tutorial of Ultra-Reliable and Low-Latency Communications in 6G:
  Integrating Theoretical Knowledge into Deep Learning</p></title>
  <aug>
    <au><snm>She</snm><fnm>C</fnm></au>
    <au><snm>Sun</snm><fnm>C</fnm></au>
    <au><snm>Gu</snm><fnm>Z</fnm></au>
    <au><snm>Li</snm><fnm>Y</fnm></au>
    <au><snm>Yang</snm><fnm>C</fnm></au>
    <au><snm>Poor</snm><fnm>HV</fnm></au>
    <au><snm>Vucetic</snm><fnm>B</fnm></au>
  </aug>
  <source>arXiv preprint arXiv:2009.06010</source>
  <pubdate>2020</pubdate>
</bibl>

<bibl id="B3">
  <title><p>Artificial Intelligence for Cloud-Assisted Smart
  Factory</p></title>
  <aug>
    <au><snm>{Wan}</snm><fnm>J.</fnm></au>
    <au><snm>{Yang}</snm><fnm>J.</fnm></au>
    <au><snm>{Wang}</snm><fnm>Z.</fnm></au>
    <au><snm>{Hua}</snm><fnm>Q.</fnm></au>
  </aug>
  <source>IEEE Access</source>
  <pubdate>2018</pubdate>
  <volume>6</volume>
  <fpage>55419</fpage>
  <lpage>55430</lpage>
</bibl>

<bibl id="B4">
  <title><p>Wireless Edge Computing With Latency and Reliability
  Guarantees</p></title>
  <aug>
    <au><snm>{Elbamby}</snm><fnm>M. S.</fnm></au>
    <au><snm>{Perfecto}</snm><fnm>C.</fnm></au>
    <au><snm>{Liu}</snm><fnm>C.</fnm></au>
    <au><snm>{Park}</snm><fnm>J.</fnm></au>
    <au><snm>{Samarakoon}</snm><fnm>S.</fnm></au>
    <au><snm>{Chen}</snm><fnm>X.</fnm></au>
    <au><snm>{Bennis}</snm><fnm>M.</fnm></au>
  </aug>
  <source>Proceedings of the IEEE</source>
  <pubdate>2019</pubdate>
  <volume>107</volume>
  <issue>8</issue>
  <fpage>1717</fpage>
  <lpage>1737</lpage>
</bibl>

<bibl id="B5">
  <title><p>Cloud computing, a practical approach</p></title>
  <aug>
    <au><snm>Velte</snm><fnm>T</fnm></au>
    <au><snm>Velte</snm><fnm>A</fnm></au>
    <au><snm>Elsenpeter</snm><fnm>R</fnm></au>
  </aug>
  <publisher>USA: McGraw-Hill, Inc.</publisher>
  <pubdate>2009</pubdate>
</bibl>

<bibl id="B6">
  <title><p>Edge Computing: Vision and Challenges</p></title>
  <aug>
    <au><snm>{Shi}</snm><fnm>W.</fnm></au>
    <au><snm>{Cao}</snm><fnm>J.</fnm></au>
    <au><snm>{Zhang}</snm><fnm>Q.</fnm></au>
    <au><snm>{Li}</snm><fnm>Y.</fnm></au>
    <au><snm>{Xu}</snm><fnm>L.</fnm></au>
  </aug>
  <source>IEEE Internet of Things Journal</source>
  <pubdate>2016</pubdate>
  <volume>3</volume>
  <issue>5</issue>
  <fpage>637</fpage>
  <lpage>646</lpage>
</bibl>

<bibl id="B7">
  <title><p>Delay-optimal computation task scheduling for mobile-edge computing
  systems</p></title>
  <aug>
    <au><snm>{Liu}</snm><fnm>J.</fnm></au>
    <au><snm>{Mao}</snm><fnm>Y.</fnm></au>
    <au><snm>{Zhang}</snm><fnm>J.</fnm></au>
    <au><snm>{Letaief}</snm><fnm>K. B.</fnm></au>
  </aug>
  <source>2016 IEEE International Symposium on Information Theory
  (ISIT)</source>
  <pubdate>2016</pubdate>
  <fpage>1451</fpage>
  <lpage>1455</lpage>
</bibl>

<bibl id="B8">
  <title><p>Resource Scheduling for Delay Minimization in Multi-Server Cellular
  Edge Computing Systems</p></title>
  <aug>
    <au><snm>{Zhang}</snm><fnm>Y.</fnm></au>
    <au><snm>{Du}</snm><fnm>P.</fnm></au>
    <au><snm>{Wang}</snm><fnm>J.</fnm></au>
    <au><snm>{Ba}</snm><fnm>T.</fnm></au>
    <au><snm>{Ding}</snm><fnm>R.</fnm></au>
    <au><snm>{Xin}</snm><fnm>N.</fnm></au>
  </aug>
  <source>IEEE Access</source>
  <pubdate>2019</pubdate>
  <volume>7</volume>
  <fpage>86265</fpage>
  <lpage>86273</lpage>
</bibl>

<bibl id="B9">
  <title><p>Cooperative Task Offloading in Three-Tier Mobile Computing
  Networks: An ADMM Framework</p></title>
  <aug>
    <au><snm>{Wang}</snm><fnm>Y.</fnm></au>
    <au><snm>{Tao}</snm><fnm>X.</fnm></au>
    <au><snm>{Zhang}</snm><fnm>X.</fnm></au>
    <au><snm>{Zhang}</snm><fnm>P.</fnm></au>
    <au><snm>{Hou}</snm><fnm>Y. T.</fnm></au>
  </aug>
  <source>IEEE Transactions on Vehicular Technology</source>
  <pubdate>2019</pubdate>
  <volume>68</volume>
  <issue>3</issue>
  <fpage>2763</fpage>
  <lpage>2776</lpage>
</bibl>

<bibl id="B10">
  <title><p>Adaptive Learning-Based Task Offloading for Vehicular Edge
  Computing Systems</p></title>
  <aug>
    <au><snm>{Sun}</snm><fnm>Y.</fnm></au>
    <au><snm>{Guo}</snm><fnm>X.</fnm></au>
    <au><snm>{Song}</snm><fnm>J.</fnm></au>
    <au><snm>{Zhou}</snm><fnm>S.</fnm></au>
    <au><snm>{Jiang}</snm><fnm>Z.</fnm></au>
    <au><snm>{Liu}</snm><fnm>X.</fnm></au>
    <au><snm>{Niu}</snm><fnm>Z.</fnm></au>
  </aug>
  <source>IEEE Transactions on Vehicular Technology</source>
  <pubdate>2019</pubdate>
  <volume>68</volume>
  <issue>4</issue>
  <fpage>3061</fpage>
  <lpage>3074</lpage>
</bibl>

<bibl id="B11">
  <title><p>Ultra-low latency cloud-fog computing for industrial Internet of
  Things</p></title>
  <aug>
    <au><snm>{Shi}</snm><fnm>C.</fnm></au>
    <au><snm>{Ren}</snm><fnm>Z.</fnm></au>
    <au><snm>{Yang}</snm><fnm>K.</fnm></au>
    <au><snm>{Chen}</snm><fnm>C.</fnm></au>
    <au><snm>{Zhang}</snm><fnm>H.</fnm></au>
    <au><snm>{Xiao}</snm><fnm>Y.</fnm></au>
    <au><snm>{Hou}</snm><fnm>X.</fnm></au>
  </aug>
  <source>2018 IEEE Wireless Communications and Networking Conference
  (WCNC)</source>
  <pubdate>2018</pubdate>
  <fpage>1</fpage>
  <lpage>6</lpage>
</bibl>

<bibl id="B12">
  <title><p>Dynamic Task Offloading and Resource Allocation for Ultra-Reliable
  Low-Latency Edge Computing</p></title>
  <aug>
    <au><snm>{Liu}</snm><fnm>C.</fnm></au>
    <au><snm>{Bennis}</snm><fnm>M.</fnm></au>
    <au><snm>{Debbah}</snm><fnm>M.</fnm></au>
    <au><snm>{Poor}</snm><fnm>H. V.</fnm></au>
  </aug>
  <source>IEEE Transactions on Communications</source>
  <pubdate>2019</pubdate>
  <volume>67</volume>
  <issue>6</issue>
  <fpage>4132</fpage>
  <lpage>4150</lpage>
</bibl>

<bibl id="B13">
  <title><p>Reliability-Optimal Offloading in Multi-Server Edge Computing
  Networks with Transmissions Carried by Finite Blocklength Codes</p></title>
  <aug>
    <au><snm>{Zhu}</snm><fnm>Y.</fnm></au>
    <au><snm>{Hu}</snm><fnm>Y.</fnm></au>
    <au><snm>{Yang}</snm><fnm>T.</fnm></au>
    <au><snm>{Schmeink}</snm><fnm>A.</fnm></au>
  </aug>
  <source>2019 IEEE International Conference on Communications Workshops (ICC
  Workshops)</source>
  <pubdate>2019</pubdate>
  <fpage>1</fpage>
  <lpage>6</lpage>
</bibl>

<bibl id="B14">
  <title><p>Risk-Sensitive Task Fetching and Offloading for Vehicular Edge
  Computing</p></title>
  <aug>
    <au><snm>{Batewela}</snm><fnm>S.</fnm></au>
    <au><snm>{Liu}</snm><fnm>C.</fnm></au>
    <au><snm>{Bennis}</snm><fnm>M.</fnm></au>
    <au><snm>{Suraweera}</snm><fnm>H. A.</fnm></au>
    <au><snm>{Hong}</snm><fnm>C. S.</fnm></au>
  </aug>
  <source>IEEE Communications Letters</source>
  <pubdate>2020</pubdate>
  <volume>24</volume>
  <issue>3</issue>
  <fpage>617</fpage>
  <lpage>621</lpage>
</bibl>

<bibl id="B15">
  <title><p>Extreme value theory: an introduction</p></title>
  <aug>
    <au><snm>De Haan</snm><fnm>L</fnm></au>
    <au><snm>Ferreira</snm><fnm>A</fnm></au>
  </aug>
  <publisher>New York: Springer Science \& Business Media</publisher>
  <pubdate>2007</pubdate>
</bibl>

<bibl id="B16">
  <title><p>Risk-sensitive reinforcement learning</p></title>
  <aug>
    <au><snm>Mihatsch</snm><fnm>O</fnm></au>
    <au><snm>Neuneier</snm><fnm>R</fnm></au>
  </aug>
  <source>Machine learning</source>
  <publisher>New York: Springer</publisher>
  <pubdate>2002</pubdate>
  <volume>49</volume>
  <issue>2-3</issue>
  <fpage>267</fpage>
  <lpage>-290</lpage>
</bibl>

<bibl id="B17">
  <title><p>{Perspectives on Risk, Assessment and Management
  Paradigms}</p></title>
  <publisher>London: IntechOpen</publisher>
  <editor>Ali G. Hessami</editor>
  <series><title><p>Books</p></title></series>
  <pubdate>2019</pubdate>
  <issue>5615</issue>
  <url>https://ideas.repec.org/b/ito/pbooks/5615.html</url>
</bibl>

<bibl id="B18">
  <title><p>Energy-Efficient Resource Allocation for Mobile-Edge Computation
  Offloading</p></title>
  <aug>
    <au><snm>{You}</snm><fnm>C.</fnm></au>
    <au><snm>{Huang}</snm><fnm>K.</fnm></au>
    <au><snm>{Chae}</snm><fnm>H.</fnm></au>
    <au><snm>{Kim}</snm><fnm>B.</fnm></au>
  </aug>
  <source>IEEE Transactions on Wireless Communications</source>
  <pubdate>2017</pubdate>
  <volume>16</volume>
  <issue>3</issue>
  <fpage>1397</fpage>
  <lpage>1411</lpage>
</bibl>

<bibl id="B19">
  <title><p>Data networks</p></title>
  <aug>
    <au><snm>Bertsekas</snm><fnm>DP</fnm></au>
    <au><snm>Gallager</snm><fnm>RG</fnm></au>
    <au><snm>Humblet</snm><fnm>P</fnm></au>
  </aug>
  <publisher>New Jersey: Prentice-Hall International</publisher>
  <pubdate>1992</pubdate>
  <volume>2</volume>
</bibl>

<bibl id="B20">
  <title><p>Queueing systems, volume 2: Computer applications</p></title>
  <aug>
    <au><snm>Kleinrock</snm><fnm>L</fnm></au>
  </aug>
  <publisher>New York: wiley</publisher>
  <pubdate>1976</pubdate>
  <volume>66</volume>
</bibl>

<bibl id="B21">
  <title><p>The departure process from a GI/G/1 queue and its applications to
  the analysis of tandem queues</p></title>
  <aug>
    <au><snm>Bertsimas</snm><fnm>DJ</fnm></au>
    <au><snm>Nakazato</snm><fnm>D</fnm></au>
  </aug>
  <publisher>Massachusetts Institute of Technology, Operations Research
  Center</publisher>
  <pubdate>1990</pubdate>
</bibl>

<bibl id="B22">
  <title><p>Characterizing the departure process of a single server queue from
  the embedded Markov renewal process at departures</p></title>
  <aug>
    <au><snm>Yeh</snm><fnm>PC</fnm></au>
    <au><snm>Chang</snm><fnm>JF</fnm></au>
  </aug>
  <source>Queueing systems</source>
  <publisher>Springer</publisher>
  <pubdate>2000</pubdate>
  <volume>35</volume>
  <issue>1-4</issue>
  <fpage>381</fpage>
  <lpage>-395</lpage>
</bibl>

<bibl id="B23">
  <title><p>Value at risk</p></title>
  <aug>
    <au><snm>Jorion</snm><fnm>P</fnm></au>
  </aug>
  <publisher>McGraw-Hill Professional Publishing</publisher>
  <pubdate>2000</pubdate>
</bibl>

<bibl id="B24">
  <title><p>On the coherence of expected shortfall</p></title>
  <aug>
    <au><snm>Acerbi</snm><fnm>C</fnm></au>
    <au><snm>Tasche</snm><fnm>D</fnm></au>
  </aug>
  <source>Journal of Banking \& Finance</source>
  <publisher>Elsevier</publisher>
  <pubdate>2002</pubdate>
  <volume>26</volume>
  <issue>7</issue>
  <fpage>1487</fpage>
  <lpage>-1503</lpage>
</bibl>

<bibl id="B25">
  <title><p>On queues in heavy traffic</p></title>
  <aug>
    <au><snm>Kingman</snm><fnm>JF</fnm></au>
  </aug>
  <source>Journal of the Royal Statistical Society: Series B
  (Methodological)</source>
  <publisher>Wiley Online Library</publisher>
  <pubdate>1962</pubdate>
  <volume>24</volume>
  <issue>2</issue>
  <fpage>383</fpage>
  <lpage>-392</lpage>
</bibl>

<bibl id="B26">
  <title><p>Optimization of conditional value-at-risk</p></title>
  <aug>
    <au><snm>Rockafellar</snm><fnm>RT</fnm></au>
    <au><snm>Uryasev</snm><fnm>S</fnm></au>
    <au><cnm>others</cnm></au>
  </aug>
  <source>Journal of risk</source>
  <pubdate>2000</pubdate>
  <volume>2</volume>
  <fpage>21</fpage>
  <lpage>-42</lpage>
</bibl>

<bibl id="B27">
  <title><p>Some remarks on the value-at-risk and the conditional
  value-at-risk</p></title>
  <aug>
    <au><snm>Pflug</snm><fnm>GC</fnm></au>
  </aug>
  <source>Probabilistic constrained optimization</source>
  <publisher>Boston, MA: Springer</publisher>
  <pubdate>2000</pubdate>
  <fpage>272</fpage>
  <lpage>-281</lpage>
</bibl>

<bibl id="B28">
  <title><p>Introduction to linear optimization</p></title>
  <aug>
    <au><snm>Bertsimas</snm><fnm>D</fnm></au>
    <au><snm>Tsitsiklis</snm><fnm>JN</fnm></au>
  </aug>
  <publisher>Belmont, MA: Athena Scientific</publisher>
  <pubdate>1997</pubdate>
  <volume>6</volume>
</bibl>

<bibl id="B29">
  <title><p>Convex optimization</p></title>
  <aug>
    <au><snm>Boyd</snm><fnm>S</fnm></au>
    <au><snm>Boyd</snm><fnm>SP</fnm></au>
    <au><snm>Vandenberghe</snm><fnm>L</fnm></au>
  </aug>
  <publisher>UK: Cambridge university press</publisher>
  <pubdate>2004</pubdate>
</bibl>

<bibl id="B30">
  <title><p>An algorithm for linear programming which requires O (((m+ n) n
  2+(m+ n) 1.5 n) L) arithmetic operations</p></title>
  <aug>
    <au><snm>Vaidya</snm><fnm>PM</fnm></au>
  </aug>
  <source>Mathematical Programming</source>
  <publisher>Springer</publisher>
  <pubdate>1990</pubdate>
  <volume>47</volume>
  <issue>1-3</issue>
  <fpage>175</fpage>
  <lpage>-201</lpage>
</bibl>

<bibl id="B31">
  <title><p>A review and comparison of solvers for convex MINLP</p></title>
  <aug>
    <au><snm>Kronqvist</snm><fnm>J</fnm></au>
    <au><snm>Bernal</snm><fnm>DE</fnm></au>
    <au><snm>Lundell</snm><fnm>A</fnm></au>
    <au><snm>Grossmann</snm><fnm>IE</fnm></au>
  </aug>
  <source>Optimization and Engineering</source>
  <publisher>Springer</publisher>
  <pubdate>2019</pubdate>
  <volume>20</volume>
  <issue>2</issue>
  <fpage>397</fpage>
  <lpage>-455</lpage>
</bibl>

<bibl id="B32">
  <title><p>Bottleneck generalized assignment problems</p></title>
  <aug>
    <au><snm>Mazzola</snm><fnm>JB</fnm></au>
    <au><snm>Neebe</snm><fnm>AW</fnm></au>
  </aug>
  <source>Engineering Costs and Production Economics</source>
  <publisher>Elsevier</publisher>
  <pubdate>1988</pubdate>
  <volume>14</volume>
  <issue>1</issue>
  <fpage>61</fpage>
  <lpage>-65</lpage>
</bibl>

<bibl id="B33">
  <title><p>The bottleneck generalized assignment problem</p></title>
  <aug>
    <au><snm>Martello</snm><fnm>S</fnm></au>
    <au><snm>Toth</snm><fnm>P</fnm></au>
  </aug>
  <source>European journal of operational research</source>
  <publisher>Elsevier</publisher>
  <pubdate>1995</pubdate>
  <volume>83</volume>
  <issue>3</issue>
  <fpage>621</fpage>
  <lpage>-638</lpage>
</bibl>

<bibl id="B34">
  <title><p>Interior point approach to linear, quadratic and convex
  programming: algorithms and complexity</p></title>
  <aug>
    <au><snm>Den Hertog</snm><fnm>D</fnm></au>
  </aug>
  <publisher>Dordrecht: Springer Science Business Media</publisher>
  <pubdate>2012</pubdate>
  <volume>277</volume>
</bibl>

<bibl id="B35">
  <title><p>On the classical logarithmic barrier function method for a class of
  smooth convex programming problems</p></title>
  <aug>
    <au><snm>Den Hertog</snm><fnm>D</fnm></au>
    <au><snm>Roos</snm><fnm>C</fnm></au>
    <au><snm>Terlaky</snm><fnm>T</fnm></au>
  </aug>
  <source>Journal of Optimization Theory and Applications</source>
  <publisher>Springer</publisher>
  <pubdate>1992</pubdate>
  <volume>73</volume>
  <issue>1</issue>
  <fpage>1</fpage>
  <lpage>-25</lpage>
</bibl>

<bibl id="B36">
  <title><p>Ultra-Reliable Communications for Industrial Internet of Things:
  Design Considerations and Channel Modeling</p></title>
  <aug>
    <au><snm>{Yang}</snm><fnm>J.</fnm></au>
    <au><snm>{Ai}</snm><fnm>B.</fnm></au>
    <au><snm>{You}</snm><fnm>I.</fnm></au>
    <au><snm>{Imran}</snm><fnm>M.</fnm></au>
    <au><snm>{Wang}</snm><fnm>L.</fnm></au>
    <au><snm>{Guan}</snm><fnm>K.</fnm></au>
    <au><snm>{He}</snm><fnm>D.</fnm></au>
    <au><snm>{Zhong}</snm><fnm>Z.</fnm></au>
    <au><snm>{Keusgen}</snm><fnm>W.</fnm></au>
  </aug>
  <source>IEEE Network</source>
  <pubdate>2019</pubdate>
  <volume>33</volume>
  <issue>4</issue>
  <fpage>104</fpage>
  <lpage>111</lpage>
</bibl>

<bibl id="B37">
  <title><p>Modeling of the fading statistics of wireless sensor network
  channels in industrial environments</p></title>
  <aug>
    <au><snm>Olofsson</snm><fnm>T</fnm></au>
    <au><snm>Ahlen</snm><fnm>A</fnm></au>
    <au><snm>Gidlund</snm><fnm>M</fnm></au>
  </aug>
  <source>IEEE Transactions on Signal Processing</source>
  <publisher>IEEE</publisher>
  <pubdate>2016</pubdate>
  <volume>64</volume>
  <issue>12</issue>
  <fpage>3021</fpage>
  <lpage>-3034</lpage>
</bibl>

</refgrp>
} 

\end{backmatter}
\end{document}